\documentclass[a4paper,fleqn]{cas-dc}

\usepackage[numbers]{natbib}
\usepackage{mathtools, amsfonts, amssymb, amsthm}
\usepackage{shuffle}
\usepackage{graphicx}
\usepackage{adjustbox}
\usepackage{booktabs, tabularx, multirow}
\usepackage{enumitem}
\usepackage{ragged2e}
\usepackage{microtype}
\usepackage{placeins}
\usepackage[nameinlink,noabbrev,capitalize]{cleveref}
\usepackage{soul}
\setlist{nosep}

\makeatletter
\g@addto@macro\UrlBreaks{\do\/\do\_\do\-\do\.}
\makeatother

\numberwithin{equation}{section}

\newcommand{\E}{\mathbb{E}}
\newcommand{\PP}{\mathbb{P}}
\newcommand{\R}{\mathbb{R}}
\newcommand{\dd}{\mathop{}\mathrm{d}}

\newcommand{\df}{\coloneqq}
\renewcommand{\vec}[1]{\operatorname{vec}\!\left(#1\right)}

\DeclareMathOperator{\Var}{Var}
\DeclareMathOperator{\Cov}{Cov}

\let\D\dd  

\theoremstyle{plain}
\newtheorem{theorem}{Theorem}[section]
\newtheorem{proposition}[theorem]{Proposition}
\newtheorem{lemma}[theorem]{Lemma}
\newtheorem{corollary}[theorem]{Corollary}
\theoremstyle{definition}

\newtheorem{assumption}[theorem]{Assumption}
\newtheorem{example}[theorem]{Example}
\theoremstyle{remark}
\newtheorem{remark}[theorem]{Remark}

\crefname{assumption}{assumption}{assumptions}
\Crefname{assumption}{Assumption}{Assumptions}

\begin{document}

\shorttitle{Signature-Based Optimal Execution for Statistical Arbitrage}

\title[mode=title]{\texorpdfstring{Signature-Based Optimal Execution for Statistical Arbitrage\\ with Path-Dependent Trading Signals}{Signature-Based Optimal Execution for Statistical Arbitrage with Path-Dependent Trading Signals}}

\author[1,3]{Gianmarco Morbelli}[]
\ead{g.morbelli@uva.nl}

\author[2,3]{Sven Karbach}[]
\ead{sven@karbach.org}

\author[1]{Mike Derksen}[]
\ead{mike.derksen@deepbluecap.com}

\affiliation[1]{%
  organization={Deep Blue Capital B.V.},
  city={Amsterdam},
  country={The Netherlands}
}

\affiliation[2]{%
  organization={Informatics Institute, University of Amsterdam},
  addressline={LAB42, Science Park 900},
  city={Amsterdam},
  postcode={1098 XH},
  country={The Netherlands}
}

\affiliation[3]{%
  organization={Korteweg-de Vries Institute for Mathematics, University of Amsterdam},
  addressline={Science Park 105–107},
  city={Amsterdam},
  postcode={1098 XG},
  country={The Netherlands}
}

\begin{abstract}
We develop a signature-based framework for optimal execution in statistical arbitrage strategies with path-dependent predictive signals. Both the alpha process and the trading speed are modelled as linear functionals of the truncated signature of a time-augmented market path, placing signal generation and execution on the same truncated signature basis. This allows the trading rule to react to the realised history of the signal while accounting for temporary impact, inventory exposure, terminal liquidation, and approximate dollar neutrality.
The main contribution is a quadratic reduction theorem: within the class of signature-linear trading speeds, the restricted path-dependent execution problem becomes a finite-dimensional concave quadratic programme in the policy coefficients. After running synthetic experiments under a mean-reverting log-spread model, we find that the fitted policy achieves a higher return on turnover than a classical $z$-score  threshold benchmark. We show how the same workflow can be deployed on a historical equity pairs-trading backtest, where the fitted signature policy again outperforms the benchmark in accounting terms.
\end{abstract}


\maketitle
\setcounter{page}{1}
\section{Introduction}

Statistical-arbitrage strategies depend not only on the quality of their predictive signals, but also on how those signals are executed.
In liquid markets, relative-value signals often decay on the same time scale on
which temporary impact, inventory risk, and neutrality constraints become
material. A trading rule for statistical arbitrage must therefore do more than decide whether a spread is rich or cheap and also decide how aggressively to trade, how much inventory to hold, and how quickly to reduce exposure as the opportunity evolves. This paper studies the execution of path-dependent statistical arbitrage signals. The motivating example is pairs trading: a spread signal indicates which long--short position is attractive, but the realised history of the spread, the cost of trading, and the requirement to remain approximately dollar-neutral all affect the optimal execution rule. Classical threshold strategies, such as $z$-score entry and exit rules, usually separate signal generation from execution. They specify when to enter and exit a trade, but they do not solve an execution problem with impact, inventory penalties, and terminal liquidation \cite{GatevGoetzmannRouwenhorst2006Pairs,AvellanedaLee2010StatArb}. Conversely, classical optimal-execution and algorithmic-trading models study temporary,
permanent, and transient impact under linear-quadratic or convex objectives
\cite{AlmgrenChriss2000OptimalExecution,Cartea2015AlgorithmicTrading,LorenzSchied2013TransientImpact,Curato0201201}, but their canonical schedules are usually driven by clock time, current inventory, or a low-dimensional Markov state rather than by the realised path of a statistical arbitrage signal. Alpha-aware execution and data-rich microstructure models incorporate
predictors such as order-flow imbalance \cite{Kolm2023DeepOFI}, but they usually
rely on parametric or learned state representations and do not yield an
explicit matrix reduction for general path-dependent signals. The goal of this paper is to connect these viewpoints. We ask whether a path-dependent statistical arbitrage signal and an execution-aware trading rule can be represented in the same feature space, so that the resulting control problem remains computationally tractable. The answer developed here uses path signatures, a sequence of iterated integrals of an observed path, which can be read as a systematic feature dictionary for
the realised history of the market.
Signature methods in finance were introduced in
\cite{GyurkoLyons2014ExtractingSignature} and applied to execution and portfolio optimisation in
\cite{optimal_execution,double_execution_signatures}. Closest to our setting,
\cite{futter2023signaturetradingpathdependentextension} develop signature
trading strategies with path-dependent signals in a mean--variance framework; neural, controlled-differential-equation, and
kernelised signature models use this feature map for rich prediction and
simulation
\cite{Buehler2020MarketGeneratorSignatures,Kidger2020NeuralCDE,3454287.3454566,Lu2024SignatureKernelMMD,Manten2022SignatureKernelCI}.
Relative to the existing signature-execution literature, the contribution here
is to formulate an execution problem for statistical arbitrage signals rather
than for a pure liquidation or multi-execution objective. The same signature
state carries both the exogenous alpha signal and the trading speed rule, while
the objective is flexible and it can combine signal reward, temporary impact, inventory risk,
terminal liquidation, and dollar-neutrality penalties. Most importantly, we then show how all these constraints remain inside one finite-dimensional quadratic programme.

\paragraph{Our approach.}
We model both the predictive alpha signal and the trading speed as linear
functionals of the same truncated signature state.
The signal matrix determines which path features predict returns, and the execution
matrix determines how the strategy trades in response to those features.
Because both objects live on the same tensor basis, impact costs, inventory
risk, dollar-neutrality penalties, and terminal liquidation penalties reduce to
a quadratic form in the execution coefficients. 
As such, the resulting optimisation is static and we can solve for the trading speed via a standard quadratic decomposition. No dynamic programming or HJB equation is solved during calibration, and no
iterative optimisation is needed during live execution. The practical attraction is that the expensive step happens off-line.
The objective matrices are estimated from historical paths, Monte Carlo simulations, or computed analytically when suitable.
After this calibration step, the live trading speed can be obtained simply by computing the current signature features and applying one
matrix-vector multiplication with the calibrated speed matrix.
For an execution desk, this means that path-dependent adaptivity can be added
without introducing a real-time dynamic-programming layer.

\paragraph{Paper Outline.}
We organize the paper as follows: \newline
\Cref{sec:framework} introduces the signature-based execution framework. After defining the lifted information path and the truncated signature state, we model both the predictive signal and the trading speed as linear functionals of the same path-feature vector. We then state the execution objective and prove the quadratic reduction theorem, which shows that the restricted path-dependent control problem becomes a finite-dimensional concave quadratic programme in the trading speed coefficients. The section also discusses the closed-form optimiser, regularisation, and the connection with classical Almgren--Chriss and linear-quadratic execution models. \Cref{sec:synthetic_historical_results} then presents the trading experiments. We test the framework on a synthetic common-trend log-spread model with an Ornstein--Uhlenbeck spread and compare the learned signature policy with a standard $z$-score pairs trading benchmark. A historical workflow then applies the same calibration procedure to an equity pair, illustrating the practical deployment of the framework on real data.
The appendices provide the supporting derivations and diagnostics. \Cref{app:tensor explicit computation} gives the explicit tensor expansion behind the matrices of the quadratic objective. \Cref{app:shuffle_algebra} presents the equivalent shuffle-algebra formulation. \Cref{sec:level2_structure} derives the moment inputs and closed-form OU Gram blocks used to validate empirical calibration. \Cref{subsec:reduction_check} reports numerical checks of the quadratic reduction, matrix conditioning, truncation-order effects and parameter sensitivity.

\begin{table}[ht]
\centering
\caption{Summary of principal notation.}
\label{tab:notation}
\small
\setlength{\tabcolsep}{4pt}
\begin{tabularx}{\linewidth}{@{}l>{\RaggedRight\arraybackslash}X@{}}
\toprule
Symbol & Meaning \\
\midrule
$Z_t\in\R^{d_z}$ & Exogenous information (market) path \\
$P_t\in\R^n$ & Unaffected mid-prices ($n$ assets) \\
$S(Z)_{s,t}$ & Signature of $Z$ over $[s,t]$ \\
$S^{\le N}(Z)_{0,t}$ & Truncated signature at level $N$ \\
$x_t\in\R^m$ & Coordinate vector of $S^{\le N}(Z)_{0,t}$  \\
$\psi_t\in\R^k$ & Projected low-dimensional basis  \\
$y_t := \int_0^t x_u\,\dd u$ & Time-integral of signature features \\
$r_t := \int_0^t \psi_u\,\dd u$ & Time-integral of projected basis \\
$v_t\in\R^n$ & Trading speed (shares per unit time) \\
$\alpha_t\in\R^n$ & Predictive signal  \\
$K\in\R^{n\times m}$ & Signal matrix; $\alpha_t = Kx_t$ \\
$B\in\R^{n\times m}$ & Trading-speed matrix; $v_t = Bx_t$ \\
$\theta = \vec{B}\in\R^{nm}$ & Column-major vectorisation of $B$ \\
$Q_t\in\R^n$ & Inventory (cumulative position) \\
$A\in\R^{nm\times nm}$ & Curvature matrix of the reduced quadratic objective \\
$b\in\R^{nm}$ & Linear coefficient of reduced objective \\
$\tilde\Lambda\in\R^{n\times n}$ & Temporary impact matrix \\
$\Sigma\in\R^{n\times n}$ & Inventory risk (urgency) matrix  \\
$\phi\ge 0$ & Inventory risk (urgency) weight \\
$\eta\ge 0$ & Dollar-neutrality weight \\
$\gamma\ge 0$ & Terminal inventory penalty \\
$G_\psi$ & Gram matrix $\E\left[\int_0^T\psi_t\psi_t^\top\,\dd t\right]$ \\
$G_r$ & Gram matrix $\E\left[\int_0^T r_t r_t^\top\,\dd t\right]$ \\
$\mathbb{A}_t$ & L\'evy area \\
$z_t$ & Rolling spread $z$-score used in the mean-reversion signal \\
$c_\alpha$ & Signal scale converting $z_t$ into instantaneous-return units \\
$\beta$ & Hedge ratio in the spread definition \\
$\tau$ & Rolling-window length for the $z$-score  \\
$\kappa$ & OU mean-reversion speed \\
$\sigma_M,\sigma_X$ & Common-trend and log-spread volatilities in the synthetic data-generating process \\
$\lambda_{\mathrm{ridge}}$ & Ridge shift used in the empirical matrix solve \\
$M_i,R_i$ & Closed-form OU moment integrals used in $G_\psi$ and $G_r$ \\
\bottomrule
\end{tabularx}
\end{table}

\section{Signature-Based Execution Framework}
\label{sec:framework}

We formulate the statistical arbitrage execution problem directly in signature
feature space.
The same truncated signature drives both alpha generation and the execution
policy, so signal extraction and trading frictions are handled within one model
rather than in two separate layers.
The guiding idea is simple: instead of forcing the problem into a small
Markovian state, we keep the path dependence explicit but choose a feature
class whose algebra is still tractable.

\subsection{Market Path and Signature Features}
\label{sec:signature theory}

Fix a finite time horizon $T>0$.
Let $(\Omega,\mathcal F,(\mathcal F_t)_{t\in[0,T]},\PP)$ be a filtered probability space satisfying the usual conditions, with $t$ the time parameter.
We consider $n$ tradable asset processes $P^{(i)}_t,i=1,\ldots,n$, and collect all exogenous information in a continuous, $\mathbb F$-adapted path
\begin{align}\label{eq:state_path_framework}
    Z_t \in \R^{d_z},\qquad t\in[0,T].
\end{align}
The unaffected mid-prices are extracted by a fixed projection $\pi_P\colon \R^{d_z}\to \R^n$,
\begin{align}\label{eq:unaffected_price_projection}
    P_t \df \pi_P Z_t
    =
    (P_t^{(1)},P_t^{(2)},\ldots,P_t^{(n)})^\top .
\end{align}
The minimal specification is \[Z_t=(t,P_t^{(1)},P_t^{(2)},\ldots,P_t^{(n)}),\] but one may append spread factors or microstructure covariates without changing the algebraic arguments below.
In trading terms, $(Z_t)_{t\leq T}$ is the information stream to which the execution rule is allowed to react.
It is not controlled by the strategy; in particular, we work in a temporary-impact setting.

Throughout the paper, we work under the following general standing assumption.

\begin{assumption}[Geometric lift of the information path]
\label[assumption]{ass:framework_lift}
The information stream $Z=(Z_t)_{t\in[0,T]}$ is given together with a chosen $(\mathcal F_t)_{t\in[0,T]}$-adapted geometric rough-path lift
\[
    \mathbf Z
    =
    \bigl(1,\mathbf Z^{1},\ldots,\mathbf Z^{\lfloor p\rfloor}\bigr)
    \in G\Omega_p([0,T];\mathbb R^{d_z})
\]
of finite $p$-variation, for some $p\geq 1$. The first level of the lift is the increment of the information path, \[ \mathbf Z^{1}_{s,t}=Z_t-Z_s, \qquad 0\leq s\leq t\leq T, \] and, for $k=2,\ldots,\lfloor p\rfloor$, $ \mathbf Z^{k}_{s,t}\in(\mathbb R^{d_z})^{\otimes k} $ denotes the $k$-th order iterated-integral component of the chosen lift. A crucial object is the extension to any $N\geq p$ in the limit $N \rightarrow\infty$ which defines the signature of the lifted path $\mathbf Z$. We will denote the signature as $S(Z)_{s,t}$ instead of $S(\mathbf{Z})_{s,t}$ when the lift is clear from context \cite{LyonsCaruanaLevy2007}\cite{optimal_execution}. 
\end{assumption}

We stress that \Cref{ass:framework_lift} is not limiting: as noted in \cite{optimal_execution} and the references therein, many processes of interest, including continuous semimartingales and discretely observed market data, admit suitable lifts to geometric rough paths and are thus amenable to our framework.

We now proceed to introduce the basis theory behind signatures that appeared already in \Cref{ass:framework_lift}:

Let
\begin{equation}\label{eq:extended tensor algebra}
    \mathcal T((\R^{d_z}))
    :=
    \prod_{k=0}^{\infty}(\R^{d_z})^{\otimes k}
\end{equation}
denote the extended tensor algebra over $\R^{d_z}$.
For a word $I=(i_1,\dots,i_k)$ with letters in $\{1,\dots,d_z\}$, define the corresponding Stratonovich iterated integral by
\begin{equation}
S^I(Z)_{s,t}
:=
\int_{s<u_1<\cdots<u_k<t}
\circ dZ_{u_1}^{i_1}\cdots \circ dZ_{u_k}^{i_k}.
\end{equation}
The signature of $Z$ over $[s,t]$ is the tensor series in $\mathcal T((\R^{d_z}))$ given by
\[
\begin{aligned}
S(Z)_{s,t}
=
\biggl(&
1,
\int_s^t \circ dZ_u,\\
&
\int_{s<u_1<u_2<t}
\circ dZ_{u_1}\otimes \circ dZ_{u_2},
\dots
\biggr),
\end{aligned}
\]
which satisfies Chen's identity 
and the shuffle product relation
\begin{equation}\label{eq:shuffle_product}
S^I(Z)_{s,t}\cdot S^J(Z)_{s,t} = \sum_{K\in I\shuffle J} S^K(Z)_{s,t},
\end{equation}
where $I\shuffle J$ is the shuffle product between words $I$ and $J$ \cite{Lyons1998,Chevyrev_2025}.

Fix a truncation level $N\in\mathbb N$ and define the truncated tensor algebra
$$
\mathcal T^{\le N}(\R^{d_z})
:=
\bigoplus_{k=0}^N (\R^{d_z})^{\otimes k}.
$$
The truncated signature over $[0,t]$ is
$$
S^{\le N}(Z)_{0,t}
:=
\bigl(S^I(Z)_{0,t}\bigr)_{|I|\le N}
\in
\mathcal T^{\le N}(\R^{d_z}),
$$
and we write
\begin{equation}\label{eq:signature vector}
    \Phi_t \df S^{\le N}( Z)_{0,t}.
\end{equation}

After fixing a basis of $\mathcal T^{\le N}(\R^{d_z})$, we identify $\Phi_t$ with a coordinate vector
$$
x_t \in \R^m,
\qquad
m=\dim\bigl(\mathcal T^{\le N}(\R^{d_z})\bigr).
$$
Linear functionals on the truncated signature then become matrix-vector products in $x_t$.
This is the object that will play the role of the state variable in the reduced
optimisation problem.
Unlike a Markov state, $x_t$ is not postulated as a small sufficient statistic.
It is a finite-dimensional feature representation of the whole observed path
over $[0,t]$ \cite{Lyons1998,LevinLyonsNi2013LearningFromPast, Chevyrev_2025}.

\begin{assumption}[Integrability of feature moments]
\label[assumption]{ass:framework_moments}
Set $y_t:=\int_0^t x_u\,\D u.$
The moments needed below are finite:
\begin{align*}
    \hspace{-7mm}\E\left[\int_0^T \|x_t\|^2\,\D t\right]<\infty,\quad
    \E\left[\int_0^T \|P_t\|^2\|y_t\|^2\,\D t\right]<\infty.
\end{align*}
\end{assumption}

\subsection{Signal, Control, and Objective}\label{sec:signal_control_execution}

The vector $x_t$ summarizes the past trajectory of the information path $Z$
through a finite collection of signature features \cite{Chevyrev_2025,Fermanian2021Embedding}.
From this point onward, the model separates cleanly into three ingredients:
an exogenous signal $\alpha_t$, an admissible trading speed $v_t$, and the
inventory process $Q_t$ induced by that speed.
The signal and the speed use the same feature vector, but only the speed
coefficients are chosen by the optimiser.

\paragraph{Predictive signal.}
We model the short-term predictive signal (or expected instantaneous return)
of the tradable assets as a linear function of the signature features,
\begin{equation}\label{eq:alpha_framework}
\alpha_t := K x_t \in\R^n,
\end{equation}
where $K\in\R^{n\times m}$ is a parameter matrix that may be estimated from historical data.
This specification allows nonlinear dependence on the past trajectory of the market
through the signature features while preserving linearity in the model parameters.

\paragraph{Trading policy.}
We restrict admissible trading policies to the class of signature-linear controls
\begin{equation}\label{eq:control_framework}
v_t := B x_t \in\R^n,
\end{equation}
where $B\in\R^{n\times m}$ is a parameter matrix.
The process $v_t$ represents the trading speed (shares per unit time) of $n\in\mathbb{N}$ underlying assets,
and the inventory $Q_t \in \mathbb{R}^{n}$ evolves according to
\begin{align}\label{eq:inventory_def}
\D Q_t = v_t\,\D t, \quad t\geq 0.
\end{align}
This parametrization defines a finite-dimensional family of path-dependent
trading strategies.
The optimisation problem is therefore not over an arbitrary adapted process but
over the coefficient matrix $B$.
We adopt the convention that positive trading speed corresponds to buying and
negative trading speed to selling.
Therefore, all timing decisions are encoded in $B$: two paths with different
signature histories can induce different trading speeds even if the current
clock time is the same.

\begin{assumption}[Admissibility and exogeneity]\label[assumption]{ass:framework_controls}
The control $v$ is progressively measurable and satisfies
$$
\E\left[\int_0^T \|v_t\|^2 \D t\right] < \infty.
$$
Moreover, the law of the information path $Z$ does not depend on the choice of $B$, i.e.\ we work in a temporary-impact setting.
\end{assumption}
\paragraph{Trading costs.}
As stated previously, we consider temporary execution costs described by a
cross-impact matrix $\tilde{\Lambda}\in\R^{n\times n}$, such that for a given
execution speed $v_t$ we incur a cost $v_t^{\top}\tilde\Lambda v_t$ per
unit of time. We assume that the impact matrix
$\tilde\Lambda\in\R^{n\times n}$ is symmetric positive definite.

\subsection{Examples of Signals}

\begin{example}\label[example]{example:z_score}
To illustrate how \eqref{eq:alpha_framework} captures classical statistical
arbitrage signals, consider a pair of assets with prices $P_t^{(1)}$ and
$P_t^{(2)}$ and define their log spread
$S_t = \log{(P_t^{(1)})} -\beta \log{( P_t^{(2)})}$. Here, $\beta$ is the
slope coefficient of the linear regression between the returns of the two
assets\footnote{The regression slope $\beta$ is also known as the hedge
ratio.}.
A standard mean-reversion strategy trades against the normalized spread
deviation \cite{GatevGoetzmannRouwenhorst2006Pairs,AvellanedaLee2010StatArb,Krauss2017StatArbSurvey} 
\begin{align}
    z_t
    &= \frac{S_t-\mu_t}{\sigma_t}, \quad  \mu_t= \frac{1}{\tau}\int_{t-\tau}^{t} S_u\,\mathrm{d}u, \\[0.4em]
    \sigma_t
    &=
    \left(
        \frac{1}{\tau}\int_{t-\tau}^{t}
        \left(S_u-\mu_t\right)^2\,\mathrm{d}u
    \right)^{1/2},
\end{align}
producing predictive returns $\alpha_t = c_\alpha(-z_t, z_t/\beta)^\top$. The signal $z_t$ is commonly known as the $z$-score\footnote{For \Cref{ass:framework_lift} to apply with $z$ as an information channel, $\sigma_t$ must remain bounded away from zero on the trading window.} \cite{AvellanedaLee2010StatArb,CaldeiraMoura2013Pairs}.
Here, $c_\alpha > 0$ is a suitable scaling constant that converts the dimensionless
$z$-score into instantaneous-return units. 
If the $z$-score is appended explicitly to the information path,
$Z_t = (t, P_t^{(1)}, P_t^{(2)}, z_t)$,
then the $z$-score level-one coordinate is the increment
$z_t-z_0$; in the scalar-augmented vector
$(1,\Delta t,\Delta P^{(1)},\Delta P^{(2)},\Delta z)$, it appears as the
fifth entry. Thus the rolling statistic is treated as an observed
exogenous channel, not as a quantity generated by the signature of prices
alone.  The signal matrix for this choice of signal can be expressed as\footnote{Setting $z_0=0$ without loss of generality.}
$$
K
=
\begin{pmatrix}
0 & 0 & 0 & 0 & -c_\alpha & 0 & \cdots & 0 \\
0 & 0 & 0 & 0 & c_\alpha/\beta & 0 & \cdots & 0
\end{pmatrix}
\in \mathbb{R}^{2\times m}.
$$
Higher-order signature terms in $x_t$ are multiplied by zero in $K$ here,
but they appear in the trading-speed matrix $B$ and therefore influence the
execution timing.
The specification \eqref{eq:alpha_framework} thus generalises classical
mean-reversion signals while allowing additional path-dependent predictors
encoded by higher-order signature terms.
\end{example}

\begin{example}[L\'{e}vy area as a path-geometry signal]
\label[example]{rem:levy_area_signal}
Another example can be constructed using the following path feature: at truncation level $N=2$, the signature state includes the antisymmetric
second-order coordinate
\begin{equation*}
\begin{aligned}
\mathbb{A}_t = S(\mathbf{Z})_{0,t}^{(1,2)} -  S(\mathbf{Z})_{0,t}^{(2,1)} 
&=
\int_0^t \bigl(P_u^{(1)}-P_0^{(1)}\bigr)\,\circ\dd P_u^{(2)} \\
&\!-
\int_0^t\! \bigl(P_u^{(2)}-P_0^{(2)}\bigr)\,\circ\dd P_u^{(1)},
\end{aligned}
\end{equation*}
also known as the L\'evy area of the two price processes.
Geometrically, $\mathbb{A}_t$ measures the signed area enclosed by the
trajectory $(P^{(1)}, P^{(2)})$ in the price plane.
Large positive (negative) $\mathbb{A}_t$ indicates that asset $P^{(1)}$
has been systematically leading (lagging) $P^{(2)}$ over $[0,t]$.
Incorporating $\mathbb{A}_t$ into the signal matrix $K$ therefore allows the
strategy to distinguish between transient dislocations (one asset briefly moves
ahead, producing large area) and persistent divergence (both assets drift apart
simultaneously, producing little area). By antisymmetry, some paths may have
zero signed area despite non-trivial joint movement, so the L\'evy area should
be interpreted as one directional path-geometry feature rather than a complete
measure of pair dependence. For example, in \cite{segmented_levy_area}, they propose an extension of the L\'evy area that is always positive and aims to improve the strength of this signal.
This is one of the possible geometric features that the level-two signature makes available
to the policy class.
\end{example}

\subsection{Objective Framework}
We now introduce the main objective of the paper.
The trader chooses the control matrix $B$ so as to maximise
\begin{align}\label{eq:objective_framework}
J(B)
&:=
\E\Big[
\int_0^T\Big(Q_t^\top \alpha_t - v_t^\top \tilde\Lambda v_t
- \phi\,Q_t^\top \Sigma Q_t \nonumber\\
&\hspace{17mm}
- \eta (Q_t^\top P_t)^2 \Big)\D t
-\gamma\|Q_T\|^2
\Big]
\end{align}
for parameters $\phi, \eta, \gamma \ge 0$.
This objective has a natural relative-value interpretation:
\begin{itemize}
\item $Q_t^\top \alpha_t$ rewards holding inventory in directions
supported by the predictive signal, capturing the statistical arbitrage
component of the strategy.
\item $v_t^\top \tilde\Lambda v_t$ models quadratic temporary
execution costs, with temporary impact matrix $\tilde\Lambda$ possibly accommodating cross-impact between assets. By construction, $\tilde{\Lambda}$ must be positive definite, as we cannot allow for costless trades.
\item $\phi\,Q_t^\top \Sigma Q_t$, with $\Sigma\in\R^{n\times n}$ symmetric positive semidefinite, penalises inventory risk (aversion to adverse
price moves while holding a position), with $\phi\ge 0$ controlling the
urgency to reduce exposure.
Setting $\phi=0$ disables the inventory-risk term, which is appropriate
in statistical arbitrage settings where the strategy deliberately builds
and unwinds inventory in response to the signal.
\item $\eta(Q_t^\top P_t)^2$ penalises net dollar exposure,
promoting approximately dollar-neutral portfolios. 
\item $\gamma\|Q_T\|^2$ is a soft terminal liquidation penalty that discourages
residual inventory at the end of the trading horizon.
\end{itemize}\medskip
\begin{remark}[Execution-only specialisation]\label{rem:execution_only}
Setting $\alpha_t\,\dd t=\dd P_t$
turns the reward
$\int_0^T Q_t^\top\alpha_t\,\dd t$ into $\int_0^T Q_t^\top\dd P_t$, which is the
mark-to-market value of inventory accrued along the realised price path.
Integration by parts then reduces \Cref{eq:objective_framework} to a pure
execution problem with temporary impact and dollar-neutrality penalty; classical
liquidation formulations of \cite{AlmgrenChriss2000OptimalExecution,Cartea2015AlgorithmicTrading,optimal_execution}
are recovered by replacing the neutrality term with a terminal-inventory
target.
\end{remark}

\begin{remark}[Units]\label[remark]{rem:units}
We work in share-based units throughout: $Q_t$ is measured in shares,
$v_t$ in shares per unit time, and $P_t$ in currency per share. The
dollar-neutrality term $(Q_t^\top P_t)^2$ then has units of currency
squared, and the impact cost $v_t^\top\tilde\Lambda v_t$ has units of
currency per unit time after $\tilde\Lambda$ absorbs the residual
dimensional factor. The predictive signal $\alpha_t$ is the
share-normalised expected instantaneous return, i.e.\ currency per share
per unit time, so $Q_t^\top\alpha_t$ has units of currency per unit time
and integrates over $[0,T]$ to a currency reward. The reduced objective values are therefore
price-scaled scores unless this signal is converted to currency-per-share units
by multiplying by current prices.
\end{remark}

\begin{remark}
The objective in \Cref{eq:objective_framework} is a modelling choice, chosen to follow standard conventions in the optimal execution literature. The framework itself is not tied to this particular specification. Any predictive signal that can be suitably represented, or approximated, as a linear functional of the chosen signature features can be incorporated by changing the signal matrix $K$. Likewise, alternative execution penalties or trading constraints can be added whenever they preserve the quadratic structure of the problem.
\end{remark}
\subsection{Quadratic Reduction}\label{sec:quadratic reduction}

We now state the central structural result of the paper. Throughout this section we use the
\emph{column-major} vectorisation convention
$\vec{\cdot}$, which stacks columns. In particular, for
$u\in\R^m$, $w\in\R^n$, and $B\in\R^{n\times m}$,
\begin{equation}
w^\top B u = \vec{B}^\top (u\otimes w).
\end{equation}
Indeed, $w^\top B u = \sum_{i=1}^{n}\sum_{j=1}^m w_i B_{ij} u_j$ is equal to
$\vec{B}^\top (u\otimes w)$ because $(u\otimes w)$ stacks the products $w_i u_j$ in the same order as $\vec{B}$ stacks $B_{ij}$ under column-major vectorisation.

\begin{theorem}[Feature-linear quadratic reduction for signature trading speeds]\label{thm:quadratic_reduction_signature}
Under \Cref{ass:framework_lift} - \ref{ass:framework_controls}, let $x_t\in\R^m$ denote the coordinate vector of $\Phi_t=S^{\le N}(\mathbf Z)_{0,t}$ and set
$$
\alpha_t = Kx_t,\qquad v_t = Bx_t,\qquad Q_t = Q_0 + \int_0^t v_u\,\D u,
$$
with $K\in\R^{n\times m}$ fixed and $B\in\R^{n\times m}$ to be optimised. We assume that the signal matrix $K$ and the trading matrix $B$ are
deterministic, $Q_0$ is deterministic, $\tilde\Lambda$ is symmetric positive definite, $\Sigma$ is
symmetric positive semidefinite, and $\phi,\eta,\gamma\ge0$.
Let $\theta:=\vec{B}\in\R^{nm}$ be the column-major vectorisation of $B$.
Then there exist a symmetric matrix $A\in\R^{nm\times nm}$, a vector $b\in\R^{nm}$, and a scalar $c\in\R$ such that the objective \Cref{eq:objective_framework}
can be written as the finite-dimensional quadratic function
\begin{equation}\label{eq:quadratic_problem}
    J(\theta)=\theta^\top A\theta+b^\top\theta+c.
\end{equation}
Furthermore, the tensors $A,b,c$ are deterministic and given by
\begin{align}
A
&=
-\mathbb{E}\Bigg[\int_0^T
(x_t x_t^\top)\otimes \widetilde{\Lambda}\,\mathrm{d}t
\nonumber\\
&\quad
+\int_0^T
(y_t y_t^\top)\otimes
\bigl(\phi\Sigma+\eta P_tP_t^\top\bigr)\,\mathrm{d}t
\nonumber\\
&\quad
+\gamma\,
y_Ty_T^\top\otimes I_n
\Bigg],
\label{eq:A_tensor_explicit}
\\[0.3em]
b
&=
\mathbb{E}\Bigg[\int_0^T
y_t\otimes
\left(
Kx_t
-2\phi\Sigma Q_0
-2\eta P_tP_t^\top Q_0
\right)\mathrm{d}t
\nonumber\\
&\quad
-2\gamma
\bigl[
y_T\otimes Q_0
\bigr]\Bigg],
\label{eq:b_tensor_explicit}
\end{align}
and
\begin{equation}\label{eq:c_tensor_explicit}
\begin{aligned}
    c={}&
    \mathbb{E}\Bigg[
    \int_0^T
    \left(
        Q_0^\top Kx_t
        - \phi Q_0^\top \Sigma Q_0
        - \eta (Q_0^\top P_t)^2
    \right)\mathrm{d}t \\
    &\hphantom{\mathbb{E}\Bigg[}
    - \gamma\lVert Q_0\rVert^2
    \Bigg],
\end{aligned}
\end{equation}
where $I_n$ denotes the $n\times n$ identity matrix and $y_t\coloneqq\int_0^t x_u\,\D u$.
\end{theorem}

\begin{proof}
Since $\theta=\vec{B}$, $v_t = Bx_t$ and $y_t=\int_0^t x_u\,\D u$, we have
\begin{equation}
    v_t = (x_t^\top \otimes I_n)\theta \ \Rightarrow \   Q_t = Q_0 + (y_t^\top \otimes I_n)\theta.
\end{equation}
It follows that $v_t$ is linear in $\theta$ and $Q_t$ is affine in $\theta$.
Consequently, each running term in \Cref{eq:objective_framework} is at most quadratic in $\theta$: The temporary-impact term $v_t^\top\widetilde\Lambda v_t$ is quadratic in $\theta$, since $v_t$ is linear in $\theta$. The inventory-risk term $Q_t^\top\Sigma Q_t$, the dollar-neutrality term $(Q_t^\top P_t)^2$, and the terminal penalty $\lVert Q_T\rVert^2$ are also quadratic, because $Q_t$ is affine in $\theta$, while $P_t$ is exogenous from \Cref{ass:framework_controls}. Since $P_t$ does not depend on $B$, hence on $\theta$, the dollar-neutrality term remains quadratic in $\theta$. Finally, \Cref{ass:framework_moments} ensures that the corresponding coefficients are integrable, hence expectation and time integration yield deterministic $A$, $b$, and $c$, with their explicit computations displayed in \Cref{app:tensor explicit computation}.
\end{proof}

\Cref{thm:quadratic_reduction_signature} tells us that the path-dependent stochastic control problem
restricted to signature-linear trading speeds reduces to a finite-dimensional
concave quadratic maximisation in the coefficient vector $\theta$. The reduction is exact
within the fixed policy class; approximation enters only through the chosen
feature set, moment estimation, and any regularisation used in the numerical
solve.

To obtain explicit formulas, we assumed the exogenous signal is represented
in the same coordinate system,
\begin{equation}\label{eq:alpha_feature}
\alpha_t = Kx_t,
\qquad K\in\R^{n\times m}.
\end{equation}
This is the natural specification when both the forecast and the trading
speed are built from the same truncated signature state.
In this representation, the matrix $A$ collects the curvature induced by
execution costs and inventory penalties, while $b$ captures the linear reward
from the predictive signal together with the effect of the initial inventory. We now present a short lemma to ensure that, under natural sign assumptions on the objective terms, the reduced quadratic problem \Cref{eq:quadratic_problem} is concave.

\begin{lemma}[Concavity of the reduced objective]\label[lemma]{lem:concavity_reduced}
Let $J(\theta)=\theta^\top A\theta+b^\top\theta+c$ be the reduction from
\Cref{thm:quadratic_reduction_signature}.
If $\tilde\Lambda\succ0$ is symmetric, $\Sigma\succeq0$ is symmetric,
and $\phi,\eta,\gamma\ge0$, then
$J$ is concave.
\end{lemma}

\begin{proof}
We show that every term inside the expectation of the curvature matrix \Cref{eq:A_tensor_explicit} is positive semidefinite.
First, all outer products are naturally positive semidefinite: $x_t x_t^\top\succeq0, \ y_t y_t^\top\succeq0$ and $P_tP_t^\top\succeq0$. By construction,
$\widetilde\Lambda\succ0$, so
$ (x_t x_t^\top)\otimes\widetilde\Lambda\succeq0.
$
Since $\phi, \eta, \gamma\geq0$, 
$\Sigma\succeq0$, we have $\phi\Sigma\succeq0$ and
$\eta P_tP_t^\top\succeq0$. It follows that the Kronecker blocks $
    (y_t y_t^\top)\otimes
    \bigl(\phi\Sigma+\eta P_tP_t^\top\bigr)
$ and 
$
    \gamma\,y_Ty_T^\top\otimes I_n
$
are positive semidefinite.
Time integration and expectation preserve positive semidefiniteness.
Since the matrix $A$ is the negative of these positive semidefinite blocks, we obtain $A\preceq0$. Strict concavity holds whenever the penalty contributions are
jointly coercive in $\theta$ (e.g.\ when $\tilde\Lambda\succ 0$ and the
empirical Gram matrix $\E \left[\int_0^T x_tx_t^\top\,\dd t\right]$ has full rank).
\end{proof}

\begin{lemma}[Closed-form optimiser]\label[lemma]{lem:theta_star}
Assume $A$ in \Cref{eq:A_tensor_explicit} is symmetric negative
definite. Then the reduced problem is strictly concave and its unique
unconstrained maximiser is
\begin{equation}\label{eq:theta_star}
\theta^\ast = -\frac12 A^{-1}b.
\end{equation}
\end{lemma}
\begin{proof}
The gradient is $\nabla_\theta J(\theta)=2A\theta+b$. Setting it to zero
gives \Cref{eq:theta_star}, and negative definiteness of $A$ yields
strict concavity.
\end{proof}

\begin{corollary}[Ridge-regularised optimiser]\label[corollary]{cor:ridge_optimizer}
\begingroup
Let $A\preceq0$ be the curvature matrix in
\Cref{thm:quadratic_reduction_signature} and fix $\rho>0$. The ridge-shifted
objective
\[
J_\rho(\theta)
=
\theta^\top(A-\rho I)\theta+b^\top\theta+c
\]
is strictly concave and has the unique unconstrained maximiser
\[
\theta^\ast_\rho
=
-\frac12(A-\rho I)^{-1}b.
\]
\endgroup
\end{corollary}

\begin{proof}
\begingroup
Since $A\preceq0$ and $\rho>0$, $A-\rho I\prec0$. Applying
\Cref{lem:theta_star} to the shifted curvature matrix gives the stated
maximiser.
\endgroup
\end{proof}

\begin{remark}[Strict vs.\ semidefinite curvature]\label{rem:strict_vs_semi}
\Cref{lem:concavity_reduced} only requires $A\preceq 0$, while
\Cref{lem:theta_star} additionally requires $A\prec 0$. The impact block alone
equals $\E\left [\int_0^T x_tx_t^\top\,\dd t\otimes\tilde\Lambda\right]$, so
$A\prec0$ holds whenever the feature Gram matrix has full rank and
$\tilde\Lambda\succ0$. In practice, however, the
empirically assembled $\hat A$ may be numerically
negative semidefinite because the
signature dictionary contains highly collinear coordinates and the dollar
neutrality and terminal penalties act only on the integrated coordinates
$y_t,y_T$. The optimiser is then made well posed by a small Tikhonov shift,
$-\frac12(A-\rho I)^{-1}b$ with $\rho>0$ small, as shown in \Cref{cor:ridge_optimizer}. If $A\preceq0$ but $A\not\prec0$, the unregularised unconstrained
quadratic need not have a finite maximiser: a finite maximiser exists only when
the linear term does not act on null directions of $A$, equivalently
$b\in\operatorname{Range}(-A)$.
The empirical implementation therefore solves the ridge-shifted problem in
\Cref{cor:ridge_optimizer}.
\end{remark}
Hard desk constraints can be imposed after the same reduction, because
they become finite-dimensional constraints on $\theta=\vec B$ once the
feature paths are fixed. Examples include bounds on trading speed, inventory,
gross notional or leverage, hard terminal inventory constraints, and hard or
banded dollar-neutrality constraints. The resulting problem is a constrained
concave quadratic programme; the experiments below use soft penalties instead
of these hard constraints.

Without signature lifting, non-Markovian dependence on the full path of $Z$
leads naturally to dynamic programming on an enlarged, potentially
infinite-dimensional state space.
Under the signature-linear policy class, the optimisation is carried
entirely by the finite-dimensional coefficient vector $\theta$ and the
precomputed moment tensors $A$ and $b$.
The non-Markovian path dependence is absorbed into the feature vector
$x_t$ and its integral $y_t$, not into the state space of a PDE. The path-dependent stochastic control problem has been reduced to a
\emph{static} finite-dimensional optimisation: find $\theta \in \R^{nm}$ that
maximises a concave quadratic.
The matrices $A$ and $b$ are precomputed \emph{off-line} from historical
paths (or from closed-form moment formulas where available), after which
the optimal policy is available via a single matrix solve
$\theta^* = -\frac{1}{2}A^{-1}b$.
During live trading, evaluating the policy requires only the current
signature feature vector $x_t$ and one matrix-vector multiply; no
dynamic programming, no differential equations, and no re-optimisation
are needed at execution time.
For an execution desk, this is the main operational message: model richness is
paid for during calibration, not during live execution.

\subsection{Regularization and Signature Complexity}
\label{subsec:snr_and_regularization}
Universal approximation results
\cite{LyonsCaruanaLevy2007,Chevyrev_2025,zbMATH08195857}
show that linear functionals on signatures are dense in appropriate spaces of
continuous path functionals as $N\to\infty$.
They do not by themselves guarantee good finite-sample calibration at any fixed
truncation level.
In practice, however, the dimension $m$ of the feature space grows as
$\sum_{k=0}^{N}d_z^k$ and financial data is noisy.
Increasing $N$ therefore adds theoretical signal (by capturing higher-order
geometric effects such as lead--lag curvatures) but simultaneously inflates
the variance of the empirical moment matrices.
\paragraph{Spectral view.}
In practice the exact matrices $A$ and $b$ are replaced by empirical
estimators $\hat A$ and $\hat b$.
Denoting the eigendecomposition of $-\hat A$ by eigenvalues $\lambda_i > 0$
and eigenvectors $v_i$, the estimated optimiser reads
\[
\hat\theta^\ast = \tfrac{1}{2}\sum_{i=1}^{nm}\frac{v_i^\top \hat b}{\lambda_i}\,v_i.
\]
The factor $1/\lambda_i$ amplifies estimation noise in $\hat b$.
For a higher-order signature component that contributes little independent
structural variance ($\lambda_i$ small), this amplification can generate
large, oscillating coefficients and excessive inventory turnover out of
sample.
\paragraph{Tikhonov regularization.}
A practical remedy is to add a ridge penalty $-\rho\|\theta\|^2$ to
the objective, as formalised in \Cref{cor:ridge_optimizer}, yielding
\begin{align*}
J_{\mathrm{ridge}}(\theta)
&=
\theta^\top(A-\rho I)\theta + b^\top\theta + c,\\
\theta^\ast_{\mathrm{ridge}}
&=
-\tfrac{1}{2}(A-\rho I)^{-1}b.
\end{align*}
The quadratic structure is preserved and the noise amplifier for the $i$-th
component becomes $1/(\lambda_i+\rho)$. Choosing $\rho$ on the scale of the near-null eigenvalues caps this
amplification by replacing $1/\lambda_i$ with $1/(\lambda_i+\rho)$, so it is not a hard threshold.
\paragraph{Lasso penalty for feature selection.}
To identify the minimal set of predictive signature coordinates, one can
instead apply an $L_1$ penalty:
$$
J_{\mathrm{lasso}}(\theta)
=
\theta^\top A\theta + b^\top\theta + c - \tau\|\theta\|_1.
$$
The resulting problem is concave, though not a quadratic programme in the
classical smooth sense because $\|\cdot\|_1$ is non-differentiable.
It is strictly concave only on directions where the quadratic part is strictly
negative.
Equivalently, minimising $-J_{\mathrm{lasso}}$ is a convex problem solvable by
standard solvers via, e.g., a second-order cone reformulation.
The $L_1$ penalty drives noise-dominated coefficients exactly to zero,
performing automated, data-driven depth truncation without fixing $N$
in advance.

\subsection{Connection to Classical Models}
\label{sec:connection_classical}

The framework contains classical deterministic execution schedules as
degenerate feature classes.
This is important conceptually: the signature approach should be viewed as an
extension of the usual execution toolbox, not as a competing formulation that
starts from unrelated assumptions.

\begin{proposition}[Time-only reduction to deterministic execution classes]\label{prop:recovery_ac_lq}
Consider a single asset ($n=1$).
If the feature map is restricted to time-only coordinates
$x_t=(1,t,\dots,t^N/N!)$, the reduced problem is a deterministic quadratic
programme over time-polynomial trading speeds. This is the same
linear-quadratic liquidation structure as Almgren--Chriss, restricted to the
chosen polynomial schedule class; increasing $N$ enlarges this deterministic
class and can approximate the continuous Almgren--Chriss schedule.
\end{proposition}

\begin{proof}
Restrict to $Z_t=(t)$ (time only) so that $x_t=(1,t,\dots,t^N/N!)$.
The feature process is deterministic and $y_t=\int_0^t x_u\,\dd u$ is a
vector of time polynomials.
The optimal trading speed is therefore
$$
v_t^* = \sum_{k=0}^N \beta_k t^k,
$$
a polynomial schedule in $t$.
At $N=1$ this gives an affine schedule $v_t^* = \beta_0 + \beta_1 t$; increasing
$N$ enlarges the deterministic schedule class used to approximate, rather than
identically recover at finite $N$, the corresponding Almgren--Chriss liquidation profile
\cite{AlmgrenChriss2000OptimalExecution}. The classical hard liquidation constraint $Q_T=-Q_0$ is represented
here only through the soft terminal penalty unless an explicit linear equality
constraint is added.
All moment tensors $A$ and $b$ reduce to time integrals of monomials,
so $J$ becomes the classical deterministic Almgren--Chriss objective.
\end{proof}

\begin{remark}[Relation to feedback LQ control]
For practitioners, the Almgren--Chriss case corresponds to a policy that only
reacts to time, or to time and remaining inventory in the usual deterministic
liquidation formulation.
The signature framework enlarges this policy class by allowing the trading rule
to react to richer summaries of the realised path, while preserving the same
quadratic-programme structure.
For control theorists, Riccati LQ control remains the natural formulation when
the trading rule is an affine feedback of the endogenous inventory state,
for example $v_t=u_0(t)+u_1(t)Q_t$. Such feedback laws are closely related to
the present quadratic structure, but they are not covered by the exogenous-path
reduction theorem without adding an endogenous-state extension. The difference
is architectural: LQ control optimises feedback laws through a state equation,
whereas the signature formulation fixes an exogenous path-feature dictionary
and optimises the coefficients of a static feature map.
\end{remark}

\begin{remark}[Computational comparison to LQ control]
Classical stochastic LQ control optimises over \emph{adapted} controls and
proceeds through a Riccati ordinary differential equation that must be solved
backward in time.
The signature approach instead fixes a pathwise feature class, precomputes
the moment tensors $A$ and $b$ once from data or closed-form formulas, and
then reads off $\theta^*$ from a single matrix inversion.
There is no backward integration and no requirement to specify a transition
density for the state.
The trade-off is that the policy class is restricted to signature-linear
functions. Signature universality gives density in suitable continuous
path-functional classes as the truncation depth increases, but the finite-$N$
problem remains a restricted policy optimisation.

\end{remark}

\section{Trading Experiments and Workflow Illustration}\label{sec:synthetic_historical_results}
This section applies the reduced execution framework to a controlled synthetic
mean-reversion experiment and to a historical workflow example. The
synthetic data-generating process uses two log prices that share a stochastic
market component while their log spread follows an Ornstein--Uhlenbeck process.
This experiment provides the main benchmark for studying how the signature
execution layer trades a fixed $z$-score alpha signal under trading costs,
terminal inventory pressure, and dollar-neutrality penalties. The historical
example applies the same workflow to an equity pair and illustrates how the
method behaves on a real relative-value trading problem. 

\subsection{Data-Generating Process}\label{subsec:application_dgp}
We simulate two assets whose unaffected log prices satisfy, for
$t\in[0,1]$,
\begin{align}
    \D \log{P_t^{(1)}} = \D M_t + \frac{1}{2} \D X_t \\
    \D \log{P_t^{(2)}} = \D M_t - \frac{1}{2} \D X_t
\end{align}
where
\begin{align}
    \D M_t = \mu \D t + \sigma_M \D W^M_t \\
    \D X_t = -\kappa X_t \D t + \sigma_X \D W^X_t
\end{align}
with $\D W^M_t,\D W^X_t$ standard Wiener processes, correlated
with parameter $\rho$. Thus $M_t$ is an
arithmetic Brownian motion with drift $\mu$ and volatility $\sigma_M$, while
$X_t$ is an OU process with mean-reversion speed $\kappa$ and volatility
$\sigma_X$. 

The common component lets both price levels wander, while the
relative-value signal remains tied to a stationary log spread. Indeed, the
spread between the two log prices follows
\begin{align}
    \D S_t &= \D \log{P_t^{(1)}} - \D \log{P_t^{(2)}}\\
    &= \D X_t
\end{align}
so $S_t=X_t$ up to its initial value.

\subsection{Standard Z-Score Benchmark and Updated Trading Costs}\label{subsec:standard_z_score_benchmark}
As a simple benchmark, we consider a standard $z$-score pairs-trading rule
\cite{GatevGoetzmannRouwenhorst2006Pairs,AvellanedaLee2010StatArb,Vidyamurthy2004PairsTrading}.
Let
\begin{equation}
    z_t := \frac{S_t-\mu_t}{\sigma_t}
\end{equation}
denote the usual normalised spread signal. Given an entry threshold $z_{\mathrm{e}}=2$,
the rule opens a long--short position when $|z_t|>z_{\mathrm{e}}$: if $z_t<-z_{\mathrm{e}}$, it
buys asset $1$ and sells asset $2$; if $z_t>z_{\mathrm{e}}$, it sells asset $1$
and buys asset $2$. Open trades are closed when the spread signal reverts to
the prescribed exit region, here represented by a crossing of $z_t=0$. Positions are capped at a fixed maximum investment amount (in currency value).

For both the benchmark and the signature strategy, ex-post
mark-to-market wealth is reported after a fixed proportional spread cost. A
trade in asset $i$ at speed $v_t^{(i)}$ is charged the execution price
\begin{equation}
    \tilde P_t^{(i)}
    =
    P_t^{(i)}\bigl(1+\xi\,\operatorname{sign}(v_t^{(i)})\bigr),
\end{equation}
 with $\xi>0$. We set $\xi=0.5\cdot10^{-4}$, corresponding to
half a basis point per unit of traded notional. The reduced objective itself
continues to use the quadratic temporary-impact term in
\Cref{eq:objective_framework}; the proportional spread cost is an accounting
metric used simply for more realistic benchmark comparisons. 

\begin{remark}
To compare the signature-based strategy with the $z$-score benchmark, we report the return on turnover (ROT) of each strategy. For a strategy with final cumulative profit $(PnL)_T$ and cumulative traded notional
\begin{equation}
\mathbf T\coloneqq \sum_{k} \sum_{i}
\bigl| \Delta Q_{t_k}^{i} \bigr| P_{t_k}^{i},
\end{equation}
we define 
\begin{equation}
\text{ROT}
\coloneqq
\frac{(PnL)_T}{\mathbf T}.
\end{equation}
Thus, return on turnover measures the profit generated per unit of traded notional and can be a useful metric to compare strategies that trade different amounts.

\end{remark}

\subsection{Large-Run Synthetic Illustration}\label{subsec:synthetic_results}
We consider two assets following \Cref{subsec:application_dgp}.
The alpha signal is the simple $z$-score specification
$\alpha_t=c_\alpha(-z_t,z_t/\beta)^\top$, with $c_{\alpha} =1.5$, described in
\Cref{example:z_score}. The displayed
large-run experiment estimates the
signature moments from $10000$ training paths and evaluates the fitted
coefficients on $5000$ test paths. Its process parameters are
$T=1$, $P^{(1)}_0=1$, $P^{(2)}_0=1$, $\mu=0$,
$\sigma_M=0.02$, $\sigma_X=0.02$, $\kappa=50$, and $\rho=0.3$.
The execution parameters are $N=2$, $Q_0=(0,0)^\top$,
$\tilde{\Lambda}=\operatorname{diag}(10^{-4},10^{-5})$,
$\eta=10^{-1}$, $\phi=0$, $\gamma=0.1$,$\beta=1$ and $\lambda_{\text{ridge}} =10^{-8}$.
Trading speeds are evaluated on $1000$ execution events over a $5000$-step
path grid. The rolling mean and standard deviation of the log-spread signal are
estimated over a warm-up window of $600$ grid steps.
This construction gives a controlled environment in which the
spread is genuinely mean reverting and the benchmark rule of
\Cref{subsec:standard_z_score_benchmark} is well specified.

The results are shown in \Cref{fig:synthetic_results_10k_train_5k_test}. 
The two middle panels illustrate the main difference between the two approaches. The classical $z$-score benchmark\footnote{For this simulation we used a benchmark maximum investment amount of 1000.} enters and exits positions according to fixed thresholds, whereas the signature strategy produces a continuous execution policy. Trading speeds are largest when the signal is active and then decrease as the terminal horizon approaches. The resulting price-scaled inventory paths
are approximately balanced across the two assets, which
indicates that the fitted speed respects the dollar-
neutrality penalty. The terminal penalty also has the intended effect: average inventory exposure is gradually reduced toward the end of the trading window.

The bottom panel compares cumulative mark-to-market PnL after
proportional execution costs. In this simulation, the signature strategy achieves a
higher average return on turnover than the threshold benchmark, approximately
$9$ bps versus $6$ bps.

\begin{figure*}[!t]
    \centering
    \includegraphics[width=\linewidth]{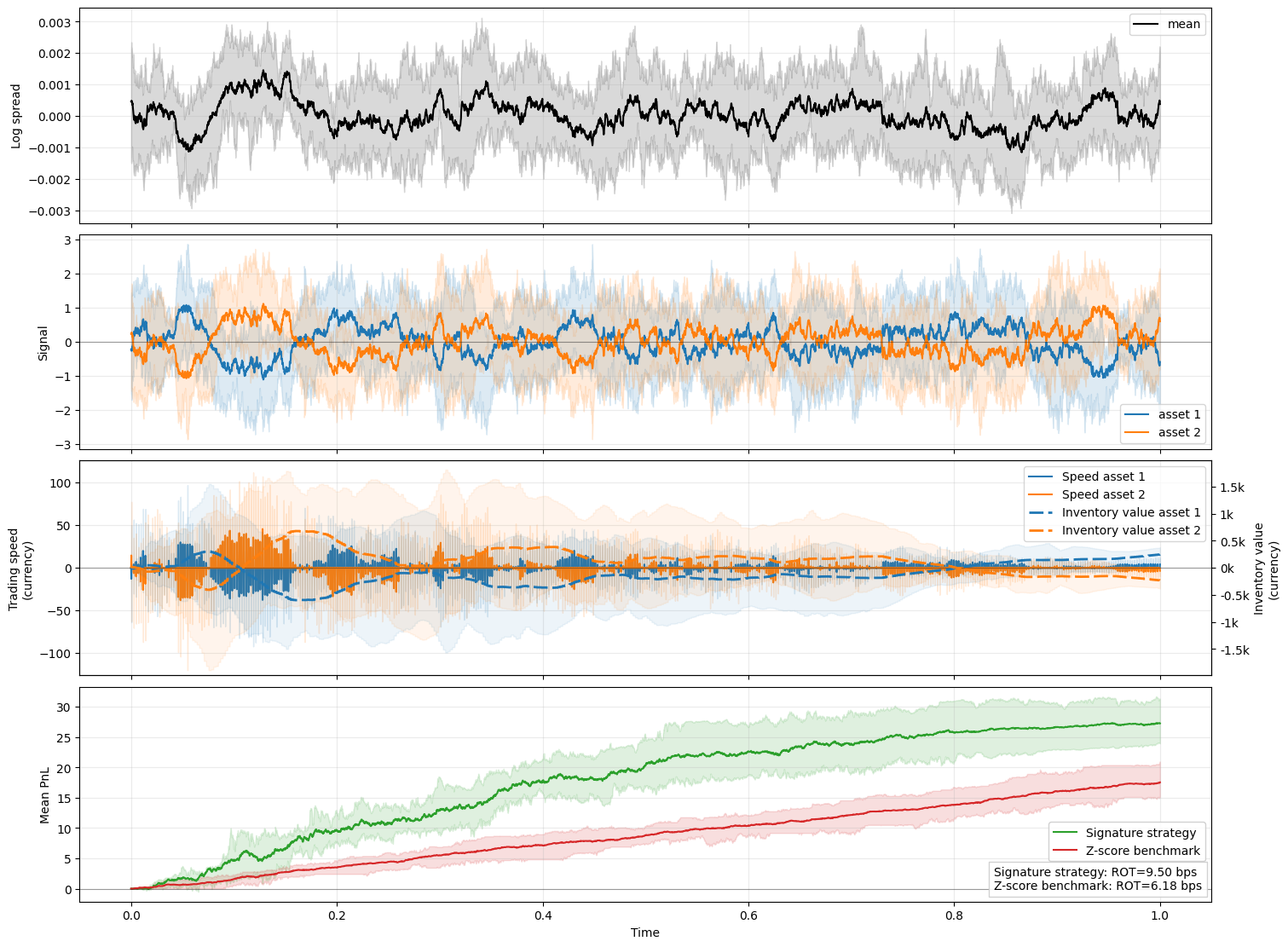}
    \caption{Synthetic benchmark against a classical
    $z$-score pairs-trading rule. The two assets follow the common-trend
    log-spread model of \Cref{subsec:application_dgp}. We display the log
    spread, the $z$-score signal, price-scaled trading speed and inventory
    paths, and cumulative PnL for the signature strategy and the benchmark of
    \Cref{subsec:standard_z_score_benchmark}. The plotted values are averages
    over $5000$ test paths after a $600$-step warm-up window used to
    estimate the rolling $z$-score; shaded regions indicate one standard
    deviation around the mean. The signature strategy uses the same $z$-score
    alpha as the benchmark but learns a continuous path-dependent execution
    rule: in this illustration
    return on turnover is higher than that of the benchmark.}
    \label{fig:synthetic_results_10k_train_5k_test}
\end{figure*}

\subsection{Historical Workflow Illustration}\label{subsec:historical_backtest}

We apply the same calibration and execution workflow to Shell
PLC\footnote{Ticker: SHEL.} and BP
PLC\footnote{Ticker: BP.}, two London Stock Exchange equities in
the energy sector for which a mean-reverting spread is economically plausible.
The empirical signature is estimated from four-day trading windows between
January 2025 and October 2025, and the fitted policy is evaluated on November
and December 2025. The execution objective uses the following parameters: $\tilde\Lambda=\operatorname{diag}(10^{-1},10^{-2})$,
$\eta=10^{-2}$, $\phi=0$, $\gamma=1$, $c_\alpha=1$,
$\beta=1$, $N=2$, and an unregularised solve ($\lambda_{\mathrm{ridge}}=0$); the
$z$-score uses an $8$-hour rolling window. Estimating the expected signature from
fixed-length past windows treats those windows as draws from a single stationary
path law, so regime drift between the training and test periods is an additional
source of out-of-sample error here. The results are displayed in
\Cref{fig:shell_bp_results}.

The top panel displays a non-stationary but visibly reverting log
spread over this test window: the spread declines during the first part of the
window, stabilises around the middle, and then rises sharply in December. The second panel shows the corresponding trading signal. The $z$-score is computed using an $8$-hour rolling window, and it can be seen how it alternates between large positive and negative regimes.

The third panel again shows the main difference between the signature solution and the benchmark. The signature strategy continuously adjusts trading speed as a function of the realised path features. This produces large but symmetrical  inventory exposures in the two assets when the signal is strong, followed by reductions in exposure when the signal weakens or the terminal horizon approaches. The policy is therefore consistent with the intended dollar-neutral execution behaviour in this window.

Finally, the bottom panel reports cumulative mark-to-market PnL.
For this particular trading window, the fitted signature policy has a significantly 
higher return on turnover,
approximately $9$ bps versus $2$ bps for the benchmark\footnote{For
this experiment we used a benchmark maximum investment amount of $500 \ 000$
British pounds (GBP) per trade.}. Given the same $z$-score signal, the quadratic model produces an adaptive execution rule that varies trade intensity and manages inventory as the realised path evolves. Furthermore, in this test window, the signature policy appears more robust than the benchmark during periods in which the signal does not consistently translate into realised returns. This is visible in November, where the $z$-score strategy records negative performance while the signature-based strategy maintains positive accounting PnL.

We stress again that our framework is general and these experiments use a rather simple signal. Further research can test the framework with richer proxies for expected returns. Natural extensions include
employing richer alpha proxies such as order-flow
imbalance features \cite{Kolm2023DeepOFI}, deep-learning microstructure
predictors \cite{SirignanoCont2019UniversalFeatures}, or residual-factor
signals in the spirit of \cite{AvellanedaLee2010StatArb}.

In conclusion, the purpose of this experiment is not to provide a systematic historical
performance claim, but to show how the signature-based execution framework can
be connected to a real equity pairs-trading problem. A full empirical validation
would require a larger universe of pairs, multiple market regimes, and a more detailed implementation of exchange fees,
taxes, borrow costs, and liquidity constraints.

\begin{figure*}[!t]
    \centering
    \includegraphics[width=\linewidth]{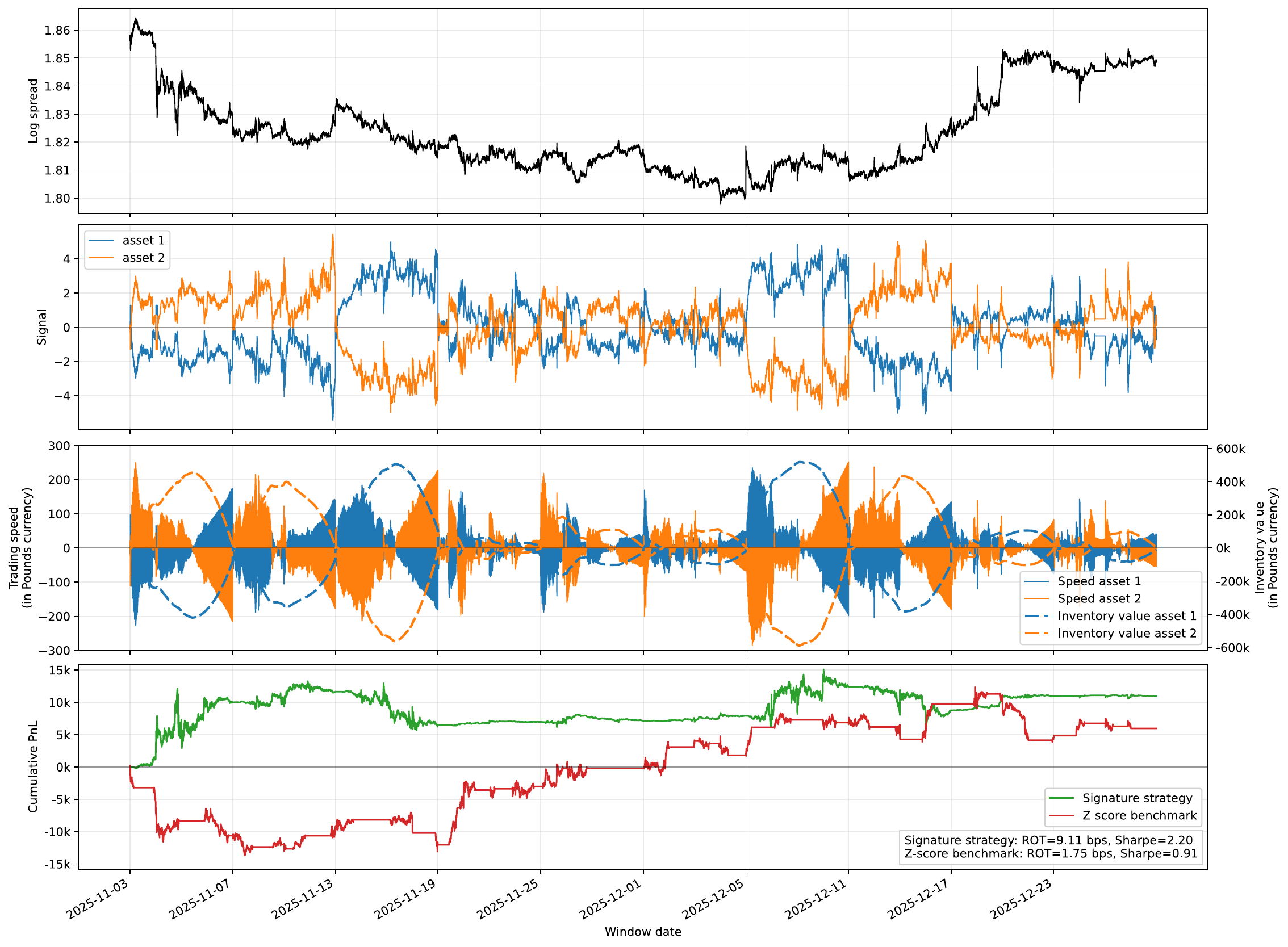}
    \caption{The historical Shell--BP deployment
    backtest on a single out-of-sample test window. The
    panels show the log spread, the $z$-score signal, price-scaled trading and
    inventory exposure in the two legs, and cumulative PnL of the signature
    strategy against the $z$-score benchmark of
    \Cref{subsec:standard_z_score_benchmark}. The signature policy is calibrated on past fixed-length trading windows and then applied on the test window. The figure illustrates how the reduced quadratic programme produces a continuous path-dependent execution rule from the same z-score alpha used by the benchmark. In this test window, the signature-based strategy outperforms the benchmark in accounting terms, as reflected by its higher return on turnover.}
    \label{fig:shell_bp_results}
\end{figure*}

\section{Conclusion}
\label{sec:conclusion}

This paper develops a signature-based framework for
signal-aware execution in statistical arbitrage. The central idea is to model
the predictive signal and the execution rule on the same path-feature space.
Path signatures provide a systematic feature map for sequential data while
preserving the algebraic structure needed to keep the execution problem
finite-dimensional.

 In the proposed framework, an exogenous information path $Z$ is
represented by its truncated signature coordinate vector $x_t$. Both the alpha
signal and the trading speed are linear functionals of this same state,
$\alpha_t=Kx_t$ and $v_t=Bx_t$. The signal matrix $K$ is fixed during
execution, while the execution matrix $B$ is optimised. Temporary impact,
inventory risk, dollar-neutrality, and terminal liquidation therefore remain
inside one finite-dimensional optimisation problem.

 The main theoretical contribution is the quadratic reduction in
\Cref{thm:quadratic_reduction_signature}. For
$\theta=\operatorname{vec}(B)$, the proposed objective restricted to
signature-linear controls can be expressed as a deterministic quadratic form in the trading speed, from which the solution is readily obtained.  During live execution, the policy is updated only through the signature features of the observed information path \Cref{eq:state_path_framework}, from which the trading speed follows by a simple matrix multiplication.
We applied the model to a synthetic simulation and to a real-data pairs-trading workflow. In both cases, the signature-based strategy achieved a higher return on turnover than the benchmark.

\paragraph{Limitations.}
The framework as presented has the following structural restrictions.
The policy class is signature-\emph{linear}: nonlinear-in-features policies
(e.g.\ neural networks over signature inputs) are outside the current quadratic
reduction, though they can be handled by approximate methods.
Execution costs are modelled as quadratic (temporary impact only); permanent
impact or order-book dynamics require separate treatment.
The Shell--BP example should not be read as a systematic empirical
validation: it uses one pair, one test window, one simple signal, and simplified
trading frictions. It also omits short-sale constraints, borrow costs, funding,
leverage caps, order-size limits, exchange fees, taxes, and latency. 

\paragraph{Future directions.}
Natural extensions include: (i) online learning of the signature policy using
rolling windows with regime-change detection; (ii) multi-asset generalisation
beyond the two-asset case studied here; (iii) depth-adaptive truncation via
$L_1$ penalisation (\Cref{subsec:snr_and_regularization}) as an alternative
to fixing $N$ in advance; (iv) hybrid analytic--empirical Gram blocks that add
area coordinates to the OU-projected basis, extending the projected-policy
calibration to the full level-two state; (v) integration with limit-order-book
microstructure models to handle permanent impact and adverse selection; and
(vi) extension to non-Markovian market impact models using the rough-path
structure of the signature.

\bibliographystyle{elsarticle-num}
\bibliography{references}

\begin{thebibliography}{10}
\expandafter\ifx\csname url\endcsname\relax
  \def\url#1{\texttt{#1}}\fi
\expandafter\ifx\csname urlprefix\endcsname\relax\def\urlprefix{URL }\fi
\expandafter\ifx\csname href\endcsname\relax
  \def\href#1#2{#2} \def\path#1{#1}\fi

\bibitem{GatevGoetzmannRouwenhorst2006Pairs}
E.~Gatev, W.~N. Goetzmann, K.~G. Rouwenhorst, Pairs trading: Performance of a relative-value arbitrage rule, The Review of Financial Studies 19~(3) (2006) 797--827.
\newblock \href {http://dx.doi.org/10.1093/rfs/hhj020} {\path{doi:10.1093/rfs/hhj020}}.

\bibitem{AvellanedaLee2010StatArb}
M.~Avellaneda, J.-H. Lee, Statistical arbitrage in the {US} equities market, Quantitative Finance 10~(7) (2010) 761--782.
\newblock \href {http://dx.doi.org/10.1080/14697680903124632} {\path{doi:10.1080/14697680903124632}}.

\bibitem{AlmgrenChriss2000OptimalExecution}
R.~Almgren, N.~Chriss, Optimal execution of portfolio transactions, Journal of Risk 3~(2) (2000) 5--39.
\newblock \href {http://dx.doi.org/10.21314/JOR.2001.041} {\path{doi:10.21314/JOR.2001.041}}.

\bibitem{Cartea2015AlgorithmicTrading}
{\'A}.~Cartea, S.~Jaimungal, J.~Penalva, Algorithmic and High-Frequency Trading, Cambridge University Press, 2015.

\bibitem{LorenzSchied2013TransientImpact}
C.~Lorenz, A.~Schied, Drift dependence of optimal trade execution strategies under transient price impact, Finance and Stochastics 17~(4) (2013) 743--770.
\newblock \href {http://dx.doi.org/10.1007/s00780-013-0211-x} {\path{doi:10.1007/s00780-013-0211-x}}.

\bibitem{Curato0201201}
G.~Curato, J.~Gatheral, F.~Lillo, Optimal execution with non-linear transient market impact, Quantitative Finance 17~(1) (2017) 41--54.
\newblock \href {http://dx.doi.org/10.1080/14697688.2016.1181274} {\path{doi:10.1080/14697688.2016.1181274}}.

\bibitem{Kolm2023DeepOFI}
P.~N. Kolm, J.~Turiel, N.~Westray, Deep order flow imbalance: Extracting alpha at multiple horizons from the limit order book, Mathematical Finance 33~(4) (2023) 1044--1081.
\newblock \href {http://dx.doi.org/10.1111/mafi.12413} {\path{doi:10.1111/mafi.12413}}.

\bibitem{GyurkoLyons2014ExtractingSignature}
L.~G. Gyurk{\'o}, T.~Lyons, M.~Kontkowski, J.~Field, Extracting information from the signature of a financial data stream, arXiv preprint arXiv:1307.7244\href {http://arxiv.org/abs/1307.7244} {\path{arXiv:1307.7244}}.

\bibitem{optimal_execution}
J.~Kalsi, T.~Lyons, I.~Perez~Arribas, Optimal execution with rough path signatures, SIAM Journal on Financial Mathematics 11~(2) (2020) 470--493.
\newblock \href {http://dx.doi.org/10.1137/19M1259778} {\path{doi:10.1137/19M1259778}}.

\bibitem{double_execution_signatures}
{\'A}.~Cartea, I.~Perez~Arribas, L.~S{\'a}nchez-Betancourt, Double-execution strategies using path signatures, SIAM Journal on Financial Mathematics 13~(4) (2022) 1379--1417.
\newblock \href {http://dx.doi.org/10.1137/21M1456467} {\path{doi:10.1137/21M1456467}}.

\bibitem{futter2023signaturetradingpathdependentextension}
O.~Futter, B.~Horvath, M.~Wiese, \href{https://arxiv.org/abs/2308.15135}{Signature trading: A path-dependent extension of the mean-variance framework with exogenous signals} (2023).
\newblock \href {http://arxiv.org/abs/2308.15135} {\path{arXiv:2308.15135}}.
\newline\urlprefix\url{https://arxiv.org/abs/2308.15135}

\bibitem{Buehler2020MarketGeneratorSignatures}
H.~Buehler, B.~Horvath, T.~Lyons, I.~Perez~Arribas, B.~Wood, Generating financial markets with signatures, SSRN Electronic Journal\href {http://dx.doi.org/10.2139/ssrn.3657366} {\path{doi:10.2139/ssrn.3657366}}.

\bibitem{Kidger2020NeuralCDE}
P.~Kidger, J.~Morrill, J.~Foster, T.~Lyons, \href{https://proceedings.neurips.cc/paper_files/paper/2020/file/4a5876b450b45371f6cfe5047ac8cd45-Paper.pdf}{Neural controlled differential equations for irregular time series}, in: Advances in Neural Information Processing Systems, Vol.~33, 2020, pp. 6696--6707.
\newline\urlprefix\url{https://proceedings.neurips.cc/paper_files/paper/2020/file/4a5876b450b45371f6cfe5047ac8cd45-Paper.pdf}

\bibitem{3454287.3454566}
P.~Bonnier, P.~Kidger, I.~P. Arribas, C.~Salvi, T.~Lyons, Deep signature transforms, Curran Associates Inc., Red Hook, NY, USA, 2019.

\bibitem{Lu2024SignatureKernelMMD}
C.-I. Lu, J.~Sester, Generative model for financial time series trained with mmd using a signature kernel, arXiv preprint arXiv:2407.19848\href {http://arxiv.org/abs/2407.19848} {\path{arXiv:2407.19848}}.

\bibitem{Manten2022SignatureKernelCI}
G.~Manten, C.~Casolo, E.~Ferrucci, S.~W. Mogensen, C.~Salvi, N.~Kilbertus, \href{https://arxiv.org/abs/2402.18477}{Signature kernel conditional independence tests in causal discovery for stochastic processes} (2025).
\newblock \href {http://arxiv.org/abs/2402.18477} {\path{arXiv:2402.18477}}.
\newline\urlprefix\url{https://arxiv.org/abs/2402.18477}

\bibitem{LyonsCaruanaLevy2007}
T.~J. Lyons, M.~Caruana, T.~L{\'e}vy, Differential Equations Driven by Rough Paths, Vol. 1908 of Lecture Notes in Mathematics, Springer, 2007.
\newblock \href {http://dx.doi.org/10.1007/978-3-540-71285-5} {\path{doi:10.1007/978-3-540-71285-5}}.

\bibitem{Lyons1998}
T.~Lyons, Differential equations driven by rough signals, Revista Matem{\'a}tica Iberoamericana 14~(2) (1998) 215--310.

\bibitem{Chevyrev_2025}
I.~Chevyrev, A.~Kormilitzin, \href{http://dx.doi.org/10.1007/978-3-031-97239-3_1}{A Primer on the Signature Method in Machine Learning}, Springer Nature Switzerland, 2025, p. 3–64.
\newblock \href {http://dx.doi.org/10.1007/978-3-031-97239-3_1} {\path{doi:10.1007/978-3-031-97239-3_1}}.
\newline\urlprefix\url{http://dx.doi.org/10.1007/978-3-031-97239-3_1}

\bibitem{LevinLyonsNi2013LearningFromPast}
D.~Levin, T.~Lyons, H.~Ni, Learning from the past, predicting the statistics for the future, learning an evolving system, arXiv preprint arXiv:1309.0260\href {http://arxiv.org/abs/1309.0260} {\path{arXiv:1309.0260}}.

\bibitem{Fermanian2021Embedding}
A.~Fermanian, Embedding and learning with signatures, Computational Statistics \& Data Analysis 157 (2021) 107148.
\newblock \href {http://dx.doi.org/10.1016/j.csda.2020.107148} {\path{doi:10.1016/j.csda.2020.107148}}.

\bibitem{Krauss2017StatArbSurvey}
C.~Krauss, Statistical arbitrage pairs trading strategies: Review and outlook, Journal of Economic Surveys 31~(2) (2017) 513--545.
\newblock \href {http://dx.doi.org/10.1111/joes.12153} {\path{doi:10.1111/joes.12153}}.

\bibitem{CaldeiraMoura2013Pairs}
J.~F. Caldeira, G.~V. Moura, Selection of a portfolio of pairs based on cointegration: A statistical arbitrage strategy, SSRN Electronic Journal\href {http://dx.doi.org/10.2139/ssrn.2196391} {\path{doi:10.2139/ssrn.2196391}}.

\bibitem{segmented_levy_area}
Z.~Guo, H.~Jin, J.~Kuang, Z.~Qian, J.~Wang, Signature decomposition method applying to pair trading, Journal of Futures Markets 46~(3) (2026) 582--603.
\newblock \href {http://dx.doi.org/https://doi.org/10.1002/fut.70075} {\path{doi:https://doi.org/10.1002/fut.70075}}.

\bibitem{zbMATH08195857}
C.~Cuchiero, P.~Schmocker, J.~Teichmann, Global universal approximation of functional input maps on weighted spaces, Constr. Approx. 63~(2) (2026) 537--612.
\newblock \href {http://dx.doi.org/10.1007/s00365-025-09726-3} {\path{doi:10.1007/s00365-025-09726-3}}.

\bibitem{Vidyamurthy2004PairsTrading}
G.~Vidyamurthy, Pairs Trading: Quantitative Methods and Analysis, John Wiley \& Sons, 2004.

\bibitem{SirignanoCont2019UniversalFeatures}
J.~Sirignano, R.~Cont, Universal features of price formation in financial markets: Perspectives from deep learning, Quantitative Finance 19~(9) (2019) 1449--1459.
\newblock \href {http://dx.doi.org/10.1080/14697688.2019.1622295} {\path{doi:10.1080/14697688.2019.1622295}}.

\bibitem{ElliottVanderHoekMalcolm2005}
R.~J. Elliott, J.~van~der Hoek, W.~P. Malcolm, Pairs trading, Quantitative Finance 5~(3) (2005) 271--276.
\newblock \href {http://dx.doi.org/10.1080/14697680500149370} {\path{doi:10.1080/14697680500149370}}.

\end{thebibliography}

\clearpage
\appendix

\section{Explicit Tensor Computation}\label[appendix]{app:tensor explicit computation}

This appendix provides detailed algebraic expansion underlying the
proof of \Cref{thm:quadratic_reduction_signature}.
We show term by term how the objective functional decomposes into a
quadratic form in $\theta = \vec{B}$, deriving the matrix $A$, vector
$b$, and scalar $c$ explicitly.

We recall that the objective functional \Cref{eq:objective_framework} reduces to
\begin{equation}
    J(\theta) = \theta^\top A\theta + b^\top\theta + c,
\end{equation}
where $\theta = \vec{B}\in\R^{nm}$ and $Q_0 \in \R^n$.
We write $v_t = Bx_t$, $y_t := \int_0^t x_u\,\dd u$, and
$\alpha_t = Kx_t$.
Under the column-major convention of \Cref{sec:quadratic reduction},
\begin{equation}
    v_t = (x_t^\top\otimes I_n)\theta,
\qquad
Q_t = Q_0 + (y_t^\top\otimes I_n)\theta.
\end{equation}
The following identity will also be useful in the following computations:
\begin{equation}\label{eq:identity tensor product}
    \bigl((y_t^\top\otimes I_n)\theta\bigr)^\top z
=
\theta^\top(y_t\otimes z),
\qquad z\in\R^n.
\end{equation}
The objective functional \Cref{eq:objective_framework} contains five groups of terms; we expand each separately and collect the contributions to $A$, $b$, and $c$.
\paragraph{Signal reward.}
\begin{align}
Q_t^\top \alpha_t
&= \left[ Q_0 + (y_t^\top\otimes I_n)\theta\right]^{\top} Kx_t \nonumber\\
&= Q_0^{\top}K x_t
+ \theta^{\top}\left(y_t \otimes Kx_t\right).
\end{align}
The first term contributes to $c$, and the second contributes to $b$.

\paragraph{Impact cost.}
\begin{align}
 v_t^\top \tilde\Lambda v_t
 &=  \left[(x_t^\top\otimes I_n)\theta\right]^{\top}
 \tilde\Lambda
 \left[(x_t^\top\otimes I_n)\theta\right] \nonumber\\
 &= \theta^{\top}\left(x_tx_t^{\top}\otimes\tilde{\Lambda}\right)\theta .
\end{align}
This contributes negatively to $A$ because impact enters the objective with a
minus sign.

\paragraph{Inventory-risk term.}
Writing $Q_t = Q_0 + (y_t^\top\otimes I_n)\theta$:
\begin{align}
 \phi\,Q_t^\top \Sigma Q_t
 &= \phi\,Q_0^\top\Sigma Q_0
   +2\phi\,\theta^\top(y_t\otimes \Sigma Q_0) \nonumber\\
 &\quad
   +\phi\,\theta^\top(y_ty_t^\top\otimes\Sigma)\theta .
\end{align}
These three terms contribute to $c$, $b$, and $A$ respectively; the quadratic
block enters $A$ with a negative sign in the objective.

\paragraph{Dollar-neutrality term.}
\begin{align}
 (Q_t^\top P_t)^2
 &= \left[Q_0^{\top}P_t + \theta^{\top} (y_t\otimes P_t)\right]^{2} \nonumber\\
 &= (Q_0^{\top}P_t)^{2}
   + \theta^{\top}(y_ty_t^{\top}\otimes P_tP_t^{\top})\theta \nonumber\\
 &\quad
   +2 (Q_0^\top P_t)\, \theta^{\top} (y_t\otimes P_t).
\end{align}
The  quadratic, linear and constant pieces contribute to $A$, $b$ and $c$ respectively,
with the quadratic block again entering $A$ negatively.

\paragraph{Terminal inventory.}
The objective contains $-\gamma Q_T^\top Q_T$; expanding:
\begin{align}
 Q_T^\top Q_T
 &=\left[ Q_0 + (y_T^\top\otimes I_n)\theta\right]^{\top}
 \left[ Q_0 + (y_T^\top\otimes I_n)\theta\right] \nonumber\\
 &= Q_0^{\top}Q_0
 +2\theta^{\top}(y_T\otimes Q_0) \nonumber\\
 &\quad
 +\theta^{\top}(y_Ty_T^{\top}\otimes I_n)\theta .
\end{align}
Multiplying by $-\gamma$:
\begin{align}
 -\gamma Q_T^\top Q_T
 &= -\gamma\,Q_0^{\top}Q_0
   -2\gamma\,\theta^{\top}(y_T\otimes Q_0) \nonumber\\
 &\quad
   -\gamma\,\theta^{\top}(y_Ty_T^{\top}\otimes I_n)\theta .
\end{align}

\section{Shuffle-Algebra Formulation}
\label[appendix]{app:shuffle_algebra}

This appendix presents an intrinsic reformulation of
\Cref{thm:quadratic_reduction_signature} in the language of shuffle
algebras, following the framework of
\cite{optimal_execution,double_execution_signatures}.
The main purpose is to connect the coordinate-based matrix presentation
of \Cref{sec:framework} to the abstract algebraic
structure underlying signatures.
Readers primarily interested in computation may skip this appendix without
loss of continuity; the key identity \eqref{eq:shuffle_product} is the
only result from this language used in the main text.

For a general introduction to signature methods, see
\cite{Chevyrev_2025}.
In this section, we reformulate the problem of \Cref{thm:quadratic_reduction_signature} in the language of shuffle algebras, restricting admissible controls to the class of \emph{signature trading speeds}.

Let $Z$ be a $d$-dimensional geometric rough path, and denote by
\[
S(Z)_{s,t}\in T\bigl((\mathbb{R}^d)\bigr)
\]
its signature over the interval $[s,t]$, where $T((\mathbb{R}^d))$ is the extended tensor algebra introduced in \eqref{eq:extended tensor algebra}. For a word
\[
I=(i_1,\dots,i_k)\in \mathcal A_d,
\]
with $\mathcal A_d$ the set of words over the alphabet $\{1,\dots,d\}$, we write $S^I(Z)_{s,t}$ for the corresponding $I$-th coordinate of the signature.

The algebra $T((\mathbb{R}^d))$ admits the natural dual space
\[
T\bigl((\mathbb{R}^d)^*\bigr),
\]
whose elements act linearly on signatures through the canonical pairing
\[
\langle \cdot,\cdot\rangle:
T\bigl((\mathbb{R}^d)^*\bigr)\times T\bigl((\mathbb{R}^d)\bigr)\to \mathbb{R}.
\]
Thus, for any $\ell\in T((\mathbb{R}^d)^*)$, the quantity
\[
\langle \ell, S(Z)_{s,t}\rangle
\]
defines a real-valued functional of the path $Z$ on $[s,t]$.
Throughout this appendix, $\ell$ denotes a \emph{finite} linear combination
of words, so that each pairing involves only finitely many signature
coordinates.

We define the class of \emph{signature trading speeds} by
\begin{equation}\label{eq:signature trading speed space}
    \mathcal T_{\mathrm{sig}}
    :=
    \left\{
    (s,t)\mapsto \langle \ell, S(Z)_{s,t}\rangle
    :
    \ell\in T\bigl((\mathbb{R}^d)^*\bigr)
    \right\}.
\end{equation}
Equivalently, if $\ell=\sum_{j} \ell_{I_j} e_{I_j},\;\; e_{I_j}\in T\bigl((\mathbb{R}^d)^*\bigr),\; \ell_{I_j}\in \mathbb{R}$ is expanded in the word basis of $T((\mathbb{R}^d)^*)$, then

$$
\langle \ell, S(Z)_{s,t}\rangle
=
\sum_j \ell_{I_j} S^{I_j}(Z)_{s,t}.
$$
As shown in \cite{optimal_execution}, this class is dense in a suitable space of admissible trading speeds, which justifies restricting the control problem to signature-based controls.

For the algebraic manipulations below, it is convenient to identify the dual tensor algebra with the space of formal linear combinations of words,
$$
T\bigl((\mathbb{R}^d)^*\bigr)\cong \mathcal W(\mathcal A_d).
$$
Under this identification, a basis element $e_I\in T((\mathbb{R}^d)^*)$ corresponds to the word $I\in\mathcal A_d$. The vector space structure is given by the usual coefficient-wise addition, while the relevant algebraic operations are the \emph{concatenation product} and the \emph{shuffle product}. In particular, if $f,g\in T((\mathbb{R}^d)^*)$, then the product of the corresponding signature functionals is represented by the shuffle product as in \Cref{sec:signature theory}:
\begin{equation}\label{eq:shuffle property}
\langle f,S(Z)_{s,t}\rangle \, \langle g,S(Z)_{s,t}\rangle
=
\langle f \shuffle g, S(Z)_{s,t}\rangle.
\end{equation}
Concatenation $f \sqcup g$ is a different operation: it adjoins letters in the
tensor basis rather than multiplying the corresponding path functionals.
For words $I$ and $J$, the concatenation $IJ$ labels the coordinate
$S^{IJ}$; the product $S^I S^J$ is instead represented by the shuffle sum in
\Cref{eq:shuffle property}.
In particular, because the path is time-augmented and the time coordinate is the
first channel, the notation
$\ell \sqcup (1)$ denotes the time-integral operation: for any speed
functional $\ell$, the integral
$\int_0^t \langle \ell, S_{0,u}\rangle\,\dd u$ is represented by
$\langle \ell \sqcup (1), S_{0,t}\rangle$.
Here the clock letter (the first channel, denoted $(1)$) is appended on the
\emph{right} of every word in $\ell$; the identity is exact for the geometric
lift because $\dd Z^{1}_u=\dd u$, so right-concatenation by $(1)$ implements
$\int_0^t(\cdot)\,\dd u$ at the level of signature coordinates.
This identification allows nonlinear expressions in the control problem to
be rewritten as linear functionals of the signature, which is the key
algebraic simplification underlying this framework.
\subsection{Dictionary for the Main Reduction}
The shuffle-algebra viewpoint is useful because it explains why products of
signature-linear quantities remain linear after moving to higher tensor levels.
The dictionary is as follows:

Let $g^i$ and $\ell^i$ be dual tensor elements such that
\begin{equation}
\alpha_t^i=\langle g^i,S_{0,t}(Z)\rangle,\qquad
v_t^i=\langle \ell^i,S_{0,t}(Z)\rangle
\end{equation}
with $i\in \{1,...,n\}$ indexing the assets. The time channel is the first
word letter and the traded price channels follow it, so the price of asset
$i$ is represented by the first-level coordinate $e_{(i+1)}$ plus its
initial value (unlike the coordinate displays of \Cref{subsec:level2_coords},
which index channels from $0$, this appendix indexes letters from $1$).
This pure signature representation assumes raw price channels; if the
signature is built from transformed or normalised inputs, the
dollar-neutrality block must instead be estimated as the mixed price-feature
moment in \Cref{eq:A_tensor_explicit}.
It follows that we can express the inventory and price process as
\begin{align}\label{eq:inventory in shuffle algebra form}
    Q^{i}_t &= Q_0^{i}+\int_0^tv_u^{i}\D u \\
         &= \langle Q_0^{i} \emptyset+ \ell^{i} \sqcup (1),S_{0,t}(Z)\rangle\\
    P_t^{i} &= \langle e_{(i+1)} + P_0^{i}\,\emptyset,S_{0,t}(Z) \rangle,
\end{align}
where $\emptyset \in \mathcal W(\mathcal{A}_d)$ denotes the empty word, i.e. with some slight abuse of notation  \[\langle \emptyset, S_{s,t}(Z)\rangle = 1.\] 
It follows that the signal term $Q_t^{\top}\alpha_t$ can be written as
\begin{align}\label{eq:signal term shuffle algebra}
    Q_t^{i}\alpha^{i}_t &= \langle Q_0^{i} \emptyset+ \ell^{i} \sqcup (1),S_{0,t}(Z)\rangle \langle g^{i},S_{0,t}(Z)\rangle\\
    &=\langle \left[Q_0^{i} \emptyset+ \ell^{i} \sqcup (1)\right]\shuffle g^{i},S_{0,t}(Z)\rangle 
\end{align}
Similarly, the temporary impact terms can be expressed as:
\begin{align}\label{eq:impact_term_shuffle_algebra}
v_t^{i}\tilde{\Lambda}^{ij}v_t^{j} &= \tilde{\Lambda}^{ij} \langle \ell^{i},S_{0,t}(Z)\rangle \langle\ell^{j},S_{0,t}(Z)\rangle \\
   &=\tilde{\Lambda}^{ij} \langle \ell^{i} \shuffle \ell^{j},S_{0,t}(Z)\rangle,
\end{align}
where we used the crucial shuffle property \Cref{eq:shuffle property}.
To ease the notation, define
\begin{equation}
    F^{i}=Q_0^{i} \emptyset+ \ell^{i} \sqcup (1).
\end{equation}
Analogously, the $\phi$ term becomes
\begin{align}\label{eq:phi_term_shuffle_algebra}
Q_t^{i}\Sigma^{ij}Q_t^{j}
&=
\Sigma^{ij}
\bigl\langle F^i,S_{0,t}(Z)\bigr\rangle
\bigl\langle F^j,S_{0,t}(Z)\bigr\rangle
\nonumber\\
&=
\Sigma^{ij}
\bigl\langle F^i \shuffle F^j,S_{0,t}(Z)\bigr\rangle .
\end{align}
The dollar neutral term in \Cref{eq:objective_framework} can be written as
\begin{align}
\mathcal N_i &:=F^i
\shuffle
\bigl[e_{(i+1)} + P_0^i \emptyset\bigr],
\label{eq:Ai_definition}
\\
Q_t^\top P_t
&=
\sum_i \langle \mathcal N_i, S_{0,t}(Z)\rangle,
\label{eq:qtpt_split}
\\
\bigl(Q_t^\top P_t\bigr)^2
&=
\Bigl\langle
\Bigl(\sum_i \mathcal N_i\Bigr)^{\shuffle 2},
\, S_{0,t}(Z)
\Bigr\rangle.
\label{eq:qtpt_squared_split}
\end{align}
The terminal inventory penalty is expressed as
\begin{align}
Q_T^{i}Q_T^{i}
&=
\bigl\langle F^i,S_{0,T}(Z)\bigr\rangle^2 \nonumber\\
&=
\bigl\langle (F^{i})^{\shuffle 2},S_{0,T}(Z)\bigr\rangle .
\end{align}
Set
\begin{equation}
\label{eq:appendix_B_definition}
\begin{aligned}
\mathcal B
&:= \sum_i F^{i} \shuffle g^i
- \sum_{i,j}\tilde{\Lambda}^{ij} \ell^i \shuffle \ell^j \\
&\quad - \phi\sum_{i,j}\Sigma^{ij} F^{i} \shuffle F^{j}
- \eta \Bigl(\sum_i \mathcal N_i\Bigr)^{\shuffle 2}.
\end{aligned}
\end{equation}
It finally follows that the objective functional \Cref{eq:objective_framework} can be rewritten as a combination of concatenations and shuffle products of linear functionals on the signature as
\begin{equation}
\label{eq:J_shuffle_form}
\begin{aligned}
J(\ell)
&=
\mathbb{E}\Biggl[
\int_0^T
\bigl\langle \mathcal B, S_{0,t}(Z)\bigr\rangle
\, \D t
- \gamma
\Bigl\langle \sum_{i}
\bigl(F^{i}\bigr)^{\shuffle 2},
\; S_{0,T}(Z)
\Bigr\rangle
\Biggr]
\\
&=
\mathbb{E}\Biggl[
\Bigl\langle
\mathcal B\sqcup (1)- \gamma \sum_{i}\bigl(F^{i}\bigr)^{\shuffle 2},
S_{0,T}(Z)
\Bigr\rangle
\Biggr]
\\
&=
\Bigl\langle
\mathcal B\sqcup (1)- \gamma \sum_{i}\bigl(F^{i}\bigr)^{\shuffle 2},
\mathbb{E}\left[S_{0,T}(Z)\right]
\Bigr\rangle ,
\end{aligned}
\end{equation}
where in the last equation we crucially moved the expectation to the signature components only, as the left-hand side of the inner product does not contain stochastic terms.
This dictionary is the algebraic content behind
\Cref{thm:quadratic_reduction_signature}: every product created by the
objective is (at most) quadratic in the policy coefficients, and its random coefficient is a signature or mixed price-feature moment that can be estimated before the policy is solved.

\begin{remark}[Truncation level required by the shuffle representation.]
If the trading speed coefficients $\ell^i$ are supported on words of
length at most $N$, then the shuffle products appearing in the
objective generally involve signature coordinates above level $N$.
For example, the terminal term $F^i\shuffle F^i$ may require levels
up to $2N+2$, while the squared dollar-neutrality term may require
levels up to $2N+4$ when prices are represented as first-level
signature coordinates. Thus the shuffle formulation is exact only when
the expected signature is available, at least, to the corresponding minimum level.
\end{remark}

\section{Moment Inputs and OU Gram Blocks}
\label[appendix]{sec:level2_structure}
This section connects the abstract quadratic reduction to the
quantities that must be estimated before solving the policy. Indeed, the quadratic reduction of \Cref{sec:quadratic reduction} leaves a concrete
calibration problem: before the policy coefficients can be solved, the moment
blocks entering $A$ and $b$ must be estimated or computed. In the notation of
\Cref{thm:quadratic_reduction_signature}, the main inputs are
\[
    \E[\int_0^T x_tx_t^\top\,\dd t],\qquad
    \E[\int_0^T y_ty_t^\top\,\dd t],\qquad
    \E[y_Ty_T^\top],
\]
together with the signal block $\E[\int_0^T y_t\otimes Kx_t\,\dd t]$ and, when
the dollar-neutrality penalty is active, the mixed price-feature block
$\E[\int_0^T y_ty_t^\top\otimes P_tP_t^\top\,\dd t]$. 

The full-signature optimiser used in the synthetic simulations of \Cref{subsec:synthetic_results} estimates these
blocks empirically from training paths. Similarly, in \Cref{subsec:historical_backtest} we computed the expected signature via past trading windows of fixed length. This section identifies a
low-dimensional part of the same calibration problem that admits closed-form
benchmarks under a standard Ornstein--Uhlenbeck spread model. Those benchmarks
are then used in \Cref{subsec:gram_validation} to check the empirical
moment-estimation pipeline. Throughout, $x_t$ denotes the full level-$N$ signature coordinate vector,
while $\psi_t = (1,t,S_t)$ denotes a projected low-dimensional basis used only for the
analytic moment calculations.

The section is organised as follows. \Cref{subsec:level2_coords} records the
signature coordinates behind the empirical moment blocks. \Cref{sec:ou_projection}
introduces the OU spread projection and derives the closed-form Gram matrix
$G_\psi$. \Cref{subsec:Gy_OU} gives the integrated-basis Gram matrix $G_r$.
\subsection{Level-N Signature Coordinates}
\label{subsec:level2_coords}
Before introducing the projected OU basis, we fix the full-signature coordinate
notation used by the empirical moment blocks. For readability, start with the
minimal two-asset price path
\begin{equation}
    Z_t = (Z_t^0,Z_t^1,Z_t^2) := (t,P_t^{(1)},P_t^{(2)}) \in \R^3
\end{equation}
and let $S^{\le N}(Z)_{0,t}$ be its truncated geometric signature of depth
$N$\footnote{We write $m=\dim(S_{0,t}^{\leq N}(\mathbf{Z}))$ for the length of
the level-$N$ signature vector.}.
The associated coordinate vector in $\R^{m}$ is
\begin{equation}\label{eq:level2_coordinates}
x_t
:=
\Big(
1,
\Delta Z_t^0,
\Delta Z_t^1,
\Delta Z_t^2,
\mathbb Z_t^{0,0},
\mathbb Z_t^{0,1},
\ldots,
\mathbb Z_t^{2,2},\ldots
\Big)^\top, 
\end{equation}
where $\Delta Z_t^i := Z_t^i-Z_0^i$ and
\begin{equation}
\mathbb Z_t^{i,j}
=
\int_{0<u_1<u_2<t}
\circ \dd Z_{u_1}^i \circ \dd Z_{u_2}^j.
\end{equation}
Thus, focusing on the first two levels, the coordinate vector contains the constant term, the three
first-level increments, and the nine second-level iterated integrals.

\begin{remark}\label{rem:dim_path_implementation}
The exposition above uses the minimal three-channel information path
$Z_t=(t,P_t^{(1)},P_t^{(2)})$. The numerical implementation in
\Cref{sec:numerics} appends the rolling $z$-score as a fourth channel,
$Z_t=(t,P_t^{(1)},P_t^{(2)},z_t)\in\R^4$ (as in \Cref{example:z_score}). The
level-two coordinate vector then has $m=1+4+16=21$ entries. All
arguments below extend channel-by-channel with no algebraic change; we keep
$d_z=3$ in the displays purely for readability.
\end{remark}

\begin{remark}[Interpretation of the level-two block]
\label{rem:level2_interpretation}
Because time is one of the channels, the shuffle identity yields
$$
\mathbb Z_t^{0,i} + \mathbb Z_t^{i,0}
=
\Delta Z_t^0\Delta Z_t^i
=
t\Delta P_t^{(i)},
\qquad i=1,2.
$$
Hence the symmetric time-price combinations recover time-weighted price
increments, while the antisymmetric combinations
$\mathbb Z_t^{0,i}-\mathbb Z_t^{i,0}$ distinguish whether price moves
occur early or late in the trading window. Likewise,
$\mathbb Z_t^{1,2}-\mathbb Z_t^{2,1}$ is the signed L\'evy area between
the two price channels and captures lead--lag information that is
invisible at level one.
\end{remark}

\subsection{Projected Bases for Analytical Moment Calculations}
\label{sec:ou_projection}

The full coordinate vector \Cref{eq:level2_coordinates} is the correct
level-$N$ state, but its moment tensors are model-specific and must
typically be estimated empirically.
For analytic insight, and to provide noise-free matrix blocks, it
is useful to work with smaller, economically interpretable bases.
We denote such reduced bases by $\psi_t$ to emphasise that they are
deliberate model reductions, not relabellings of the full signature state.
The key point is that the Gram matrices of $\psi_t$ can sometimes be
computed in closed form under a specific spread model.
In the present numerical implementation these formulas serve two purposes: they check the
empirical moment-estimation pipeline and, used
directly in the policy solve, they calibrate a projected policy with lower
sampling error than its Monte-Carlo counterpart.
Concretely, $G_\psi$ validates the impact block $\E[\int_0^T x_tx_t^\top\,\dd t]$,
while $G_r$ and $\E[r_Tr_T^\top]$ validate the inventory-risk and
terminal-liquidation blocks; the dollar-neutrality price-feature block does not
reduce to these Gram matrices and stays empirical unless the reference-price
approximation $P_tP_t^\top\approx P_0P_0^\top$ is imposed
(cf.\ \Cref{rem:neutrality_P0}).
Combining them with empirical higher-order geometric blocks is a natural
extension but is not the optimiser used for the reported full-signature trades.
Accordingly, the projected basis is an implementation device: the full theorem
remains stated for $x_t$, while $\psi_t$ isolates the part of the problem for
which exact moment calculations are especially transparent.

\paragraph{OU spread projection.}

In pairs-trading applications \cite{GatevGoetzmannRouwenhorst2006Pairs,Vidyamurthy2004PairsTrading,AvellanedaLee2010StatArb}, let $S_t$ denote a one-dimensional spread
coordinate extracted from $(P_t^{(1)},P_t^{(2)})$; the synthetic
experiment of \Cref{sec:numerics} instantiates this coordinate as the log
spread $S_t=\log P_t^{(1)}-\log P_t^{(2)}$, which equals the OU component
$X_t$ in \Cref{subsec:application_dgp}. Assume that
$S$ follows the Ornstein--Uhlenbeck dynamics \cite{ElliottVanderHoekMalcolm2005,AvellanedaLee2010StatArb}
\begin{equation}\label{eq:OU}
\dd S_t = -\kappa(S_t-\mu)\,\dd t + \sigma\,\dd W_t,
\qquad \kappa>0.
\end{equation}
Writing $a:=S_0-\mu$, its first two moments are
\[
\begin{aligned}
m(t) &:= \E[S_t] = \mu + ae^{-\kappa t}, \\
v(t) &:= \Var(S_t) = \frac{\sigma^2}{2\kappa}\bigl(1-e^{-2\kappa t}\bigr).
\end{aligned}
\]

Consider the reduced basis
\begin{equation}\label{eq:x_OU_basis}
\psi_t := (1,t,S_t)^\top \in \R^3.
\end{equation}
Its Gram matrix
\[
G_\psi := \E\int_0^T \psi_t\psi_t^\top\,\dd t
\]
is explicit. Writing
$M_0:=\int_0^T m(t)\,\dd t$,
$M_1:=\int_0^T tm(t)\,\dd t$, and
$M_2:=\int_0^T (m(t)^2+v(t))\,\dd t$, we have
\[
G_\psi =
\begin{pmatrix}
T & \frac{T^2}{2} & M_0 \\
\frac{T^2}{2} & \frac{T^3}{3} & M_1 \\
M_0 & M_1 & M_2
\end{pmatrix},
\]
with
\begin{align*}
\int_0^T m(t)\,\dd t
&=
\mu T + \frac{a}{\kappa}\bigl(1-e^{-\kappa T}\bigr), \\
\int_0^T tm(t)\,\dd t
&=
\mu\frac{T^2}{2}
+
a\Bigl(
\frac{1-e^{-\kappa T}}{\kappa^2}
- \frac{Te^{-\kappa T}}{\kappa}
\Bigr), \\
\int_0^T \E[S_t^2]\,\dd t
&=
\mu^2T
+
2\mu a\frac{1-e^{-\kappa T}}{\kappa} \nonumber\\
&\quad
+
a^2\frac{1-e^{-2\kappa T}}{2\kappa} \nonumber\\
&\quad
+ \frac{\sigma^2}{2\kappa}T
- \frac{\sigma^2}{4\kappa^2}\bigl(1-e^{-2\kappa T}\bigr).
\end{align*}

The basis \Cref{eq:x_OU_basis} is a tractable mean-reversion proxy, not the
full signature state. Its purpose is to isolate the part of the reduced
quadratic programme that already admits closed-form calibration under a standard
spread model.
In particular, $(1,t,S_t)$ captures the baseline mean-reversion economics:
time, current spread level, and the accumulated effect of carrying inventory
against that spread.

\subsection{Integrated-Basis Gram Matrix}
\label{subsec:Gy_OU}

For the projected basis \Cref{eq:x_OU_basis}, define
\begin{equation}\label{eq:y_basis_OU}
r_t
:=
\int_0^t \psi_u\,\dd u
=
\Bigl(t,\frac{t^2}{2},I_t\Bigr)^\top,
\qquad
I_t := \int_0^t S_u\,\dd u.
\end{equation}
The corresponding inventory Gram matrix is
\begin{equation}
G_r := \E\int_0^T r_t r_t^\top\,\dd t.
\end{equation}
For deterministic $S_0$, the OU covariance is
\begin{equation}\label{eq:OU_cov}
\Cov(S_t,S_s)
=
\frac{\sigma^2}{2\kappa}
\Bigl(e^{-\kappa|t-s|}-e^{-\kappa(t+s)}\Bigr),
\end{equation}
the integrated spread has mean
\begin{equation}\label{eq:EI}
\E[I_t]
=
\mu t + \frac{a}{\kappa}\bigl(1-e^{-\kappa t}\bigr)
\end{equation}
and variance
\begin{equation}\label{eq:VarI_closed}
\Var(I_t)
=
\frac{\sigma^2}{\kappa^2}t
- \frac{\sigma^2}{\kappa^3}\bigl(1-e^{-\kappa t}\bigr)
- \frac{\sigma^2}{2\kappa^3}\bigl(1-e^{-\kappa t}\bigr)^2.
\end{equation}
Hence
\begin{equation}\label{eq:EI2}
\E[I_t^2] = \bigl(\E[I_t]\bigr)^2 + \Var(I_t).
\end{equation}
The time-polynomial entries of $G_r$ are immediate,
\begin{equation}
\begin{aligned}
\int_0^T t^2\,\dd t &= \frac{T^3}{3}, \\
\int_0^T t\frac{t^2}{2}\,\dd t &= \frac{T^4}{8}, \\
\int_0^T \Bigl(\frac{t^2}{2}\Bigr)^2\dd t &= \frac{T^5}{20},
\end{aligned}
\end{equation}
while the mixed entries involving $I_t$ are obtained by setting
\begin{equation}
R_1:=\int_0^T t\,\E[I_t]\,\dd t,
\qquad
R_2:=\int_0^T \frac{t^2}{2}\E[I_t]\,\dd t,
\end{equation}
where
\begin{align*}
R_1
&=
\mu\frac{T^3}{3}
+
\frac{aT^2}{2\kappa}
- \frac{a}{\kappa}
\Bigl(
\frac{1-e^{-\kappa T}}{\kappa^2}
- \frac{Te^{-\kappa T}}{\kappa}
\Bigr), \\
R_2
&=
\mu\frac{T^4}{8}
+
\frac{aT^3}{6\kappa} \\
&\quad
- \frac{a}{2\kappa}\int_0^T t^2e^{-\kappa t}\,\dd t,
\end{align*}
with
\begin{equation}\label{eq:int_t2_exp}
\begin{aligned}
\int_0^T t^2e^{-\kappa t}\,\dd t
&=
\frac{2}{\kappa^3}\bigl(1-e^{-\kappa T}\bigr) \\
&\quad
- e^{-\kappa T}\Bigl(\frac{T^2}{\kappa}+\frac{2T}{\kappa^2}\Bigr).
\end{aligned}
\end{equation}
Finally,
\begin{equation}\label{eq:int_EI2}
\int_0^T \E[I_t^2]\,\dd t
=
\int_0^T \bigl(\E[I_t]\bigr)^2\,\dd t
+
\int_0^T \Var(I_t)\,\dd t,
\end{equation}
and the right-hand side is an elementary polynomial-exponential
expression after substituting \Cref{eq:EI,eq:VarI_closed}. Thus the
entire matrix $G_r$ is available in closed form.

The signal block $\E[\int_0^T r_t\otimes K\psi_t\,\dd t]$ entering $b$ is closed
form by the same calculation: with the projected signal proportional to the
spread $S_t$, its entries are $\int_0^T t\,m(t)\,\dd t$,
$\int_0^T \tfrac{t^2}{2}m(t)\,\dd t$ (both already available from $G_\psi$), and
$\int_0^T \E[I_tS_t]\,\dd t$, where
\begin{equation}\label{eq:cov_I_S}
\Cov(I_t,S_t)
=
\int_0^t \Cov(S_u,S_t)\,\dd u
=
\frac{\sigma^2}{2\kappa^2}\bigl(1-e^{-\kappa t}\bigr)^2,
\end{equation}
so that $\E[I_tS_t]=\E[I_t]\,m(t)+\Cov(I_t,S_t)$ is again elementary after
substituting \Cref{eq:EI}.

\begin{remark}[Neutrality moment under OU projection]
\label[remark]{rem:neutrality_P0}
The dollar-neutrality penalty contributes a block
$(y_t y_t^\top)\otimes \eta P_tP_t^\top$ to the curvature matrix
(cf.\ \Cref{eq:A_tensor_explicit}), which requires the mixed moment
$\E[\int_0^T r_tr_t^\top \otimes P_tP_t^\top\,\dd t]$. Under the OU spread
projection the price vector $P_t$ is not a function of $\psi_t$, so this
mixed moment does not reduce to the OU Gram blocks. In a projected analytic
approximation one may replace
$P_tP_t^\top$ by its reference value $P_0P_0^\top$ inside the
neutrality block, reducing it to $\eta\bigl(P_0P_0^\top\bigr)\otimes G_r$.
This is a reference-price approximation: it evaluates the dollar-neutrality
penalty around a fixed price scale and should not be interpreted as a general
upper or lower bound. In the
empirical construction no such approximation is required: $P_tP_t^\top$
is observed on each simulated path and the block is estimated directly.
The reported canonical optimiser uses this empirical construction.
\end{remark}

\subsection{Validation and Diagnostics of the Framework}
\label{sec:numerics}

The theoretical framework of
\Cref{sec:framework,sec:level2_structure} is now validated on the same
common-trend log-spread pair introduced in \Cref{subsec:application_dgp}.
The diagnostics serve three purposes: to check that the analytical
Gram matrices of \Cref{sec:ou_projection,subsec:Gy_OU} agree with Monte
Carlo estimates, to visualise the optimised policy on representative paths,
and to quantify how the reduced objective and accounting diagnostics change
when moving from the level-one ($N=1$) to the level-two ($N=2$)
signature specification.
Robustness is assessed through signal-strength, ridge, and sample-size sweeps
and through independent train/test splits, rather than through any claim of
universal level-two dominance.
The reported synthetic workflow is:
\begin{enumerate}[label=(\roman*)]
    \item simulate independent training paths and compute the time-augmented
    signature features $x_t$ and their integrals $y_t$;
    \item assemble the empirical curvature matrix and linear term
    $(\hat A,\hat b)$ using the discrete version of
    \Cref{eq:A_tensor_explicit,eq:b_tensor_explicit};
    \item symmetrise $\hat A$, apply the fixed ridge shift
    $\lambda_{\mathrm{ridge}}$, and solve
    $\hat\theta_\rho=-\frac12(\hat A-\lambda_{\mathrm{ridge}}I)^{-1}\hat b$;
    \item evaluate the fitted policy $v_t=\hat B_\rho x_t$ on held-out paths
    under the same $\dd t$-weighted reduced objective and separately record
    mark-to-market accounting metrics.
\end{enumerate}

All reported policies are estimated from $1000$ Monte Carlo training paths and
evaluated on $1000$ held-out paths, with $1000$ execution orders placed
uniformly over $[0,T]$ on each test path.

The signatures are computed from the discrete piecewise-linear interpolants of
the simulated paths.
The matrix assembly integrates against $\dd t$ over the full simulator grid,
while the backtest evaluates the fitted speed on $1000$ execution buckets with
quadrature weights equal to the bucket widths, which sum to the trading horizon
(see the units convention in \Cref{rem:units}). The trading speed is thus
treated as a rate throughout, so the policy is evaluated under the same
$\dd t$-weighted criterion it was calibrated on.
Features are evaluated at bucket endpoints in both the matrix assembly and
the backtest. This gives a consistent discretisation of the continuous-time
formula, but a strictly predictable live implementation would use left-endpoint
features. The difference is an implementation robustness issue rather than part
of the algebraic reduction. The
price channels $P^{(1)},P^{(2)}$ enter both the signature and the
dollar-neutrality block as raw mid-prices (identity transform), so there is no
scaled-versus-raw price mismatch, and the neutrality block is assembled as the
empirical mixed price-feature moment
$\E[\int_0^T y_ty_t^\top\otimes P_tP_t^\top\,\dd t]$ on each path. The
dollar-neutrality, inventory-risk, and terminal-liquidation terms are soft
penalties rather than desk-enforceable hard constraints; hard constraints are
left to future work.
Thus the experiment is a discrete approximation of the continuous-time
formulas rather than a separate continuous-time solver.
The main numerical findings are threefold.
First, the reduced matrix problem is well behaved under the canonical
calibration: the estimated quadratic form has the expected sign structure, and
the policy can be computed by a single regularised linear solve.
Second, the fitted policy has the expected relative-value structure: it opens
temporary long--short positions against spread dislocations and then reduces
inventory before the terminal horizon.
Third, level-two signature information raises the reduced objective relative
to level one in this synthetic mean-reversion experiment, but the objective gain is tied
to greater turnover and larger temporary inventory exposure.

\subsection{\texorpdfstring{Common-Trend Log-Spread Pair}{Common-Trend Log-Spread Pair}}

The first $\tau$ execution events on each path are masked to zero
so that the policy only trades once the rolling-window estimate of the
$z$-score is fully populated. The baseline path parameters are
\begin{equation}\label{eq:path_parameters_base_set}
\begin{aligned}
\kappa &= 50, & \mu &= 0, & \sigma_M &= 0.02, & \sigma_X &= 0.02, \\
\rho &= 0.3, & \tau &= 100, & \beta &= 1, \\
T &= 1, & P_0^{(1)} &= 30, & P_0^{(2)} &= 5.
\end{aligned}
\end{equation}
For the execution objective, the calibration uses
\begin{equation}\label{eq:solver_parameters_base_set}
\begin{aligned}
\Lambda_1 &= 3\times10^{-3}, & \Lambda_2 &= 3\times10^{-4}, & Q_0 &= (0,0)^\top, \\
\eta &= 10^{-2}, & \phi &= 0, \\
\gamma &= 0.1, & \lambda_{\mathrm{ridge}} &= 10^{-8}.
\end{aligned}
\end{equation}
The canonical impact matrix is diagonal:
$\Lambda_{12}=\Lambda_{21}=0$, so that
$\tilde\Lambda=\operatorname{diag}(\Lambda_1,\Lambda_2)$ in the notation of
\Cref{eq:objective_framework}.
The signal scale is
$c_\alpha=\tfrac12\,\kappa\,\widehat{\operatorname{sd}}(S)\approx5.02\times10^{-2}$,
half the OU reversion drift per unit of spread standard deviation, estimated
from the training paths and recorded in the run manifest.
Cross-impact is part of the model and code interface, but the headline
mean-reverting-pair experiment does not exercise it.
With $\phi=0$, the inventory-risk penalty is absent; the strategy is
driven by the spread signal, temporary impact, dollar-neutrality, and the
terminal inventory penalty. This choice isolates the statistical arbitrage
mechanism: the strategy is allowed to build a temporary long--short position
when the spread signal is strong, but it still pays impact costs and is
penalised for ending the trading window with residual inventory. The simulated market fixes the economic object traded in the rest
of the section: the two assets share a common stochastic log-price trend, their
log spread mean reverts around zero, and the rolling $z$-score supplies the
dimensionless dislocation signal that is converted into alpha by $c_\alpha$.
The following diagnostics and metric-specific comparisons refer to this baseline
market.

\subsection{Analytical versus Empirical Gram Matrices}
\label{subsec:gram_validation}

A key practical advantage of the framework is that the Gram matrices
$G_\psi$ and $G_r$ entering the OU-projected basis (see
\Cref{sec:ou_projection,subsec:Gy_OU}) are available in closed form under the
continuous-time OU spread model.
We verify the implementation and quantify the noise reduction by
comparing the analytical formulas against Monte Carlo estimates as a
function of the sample size $M$.

\Cref{fig:gram_convergence} shows the Frobenius-norm errors
$\|G^{\mathrm{emp}}_\psi(M) - G^{\mathrm{exact}}_\psi\|_F$ and
$\|G^{\mathrm{emp}}_r(M) - G^{\mathrm{exact}}_r\|_F$ for both projected blocks,
together with the standard $O(M^{-1/2})$ Monte Carlo reference slope.
The targets $G_\psi^{\mathrm{exact}}$ and $G_r^{\mathrm{exact}}$ are the
closed-form OU moment integrands of \Cref{sec:ou_projection,subsec:Gy_OU}
evaluated on the \emph{same} simulator time grid as the estimator. Matching the
grid cancels the deterministic time-discretisation error of the Riemann sum
(for example $\int_0^T t\,\dd t=\tfrac12$ versus the right-point value
$\tfrac12+\tfrac{1}{2n_{\mathrm{steps}}}$), so the residual is pure Monte Carlo
sampling error and decays at the $O(M^{-1/2})$ rate; at each sample size the
error is averaged over random resamples of the held-out paths.
The empirical estimator retains non-zero sampling error at finite sample
sizes, whereas the closed-form OU formula removes this source of
estimation noise for the projected OU blocks. This is a calibration diagnostic:
it shows that part of the reduced matrix has a deterministic benchmark under
the projected OU model.
The reported optimiser itself is still assembled empirically from the simulated
full-signature paths.

\begin{figure}
\includegraphics[width=0.85\linewidth]{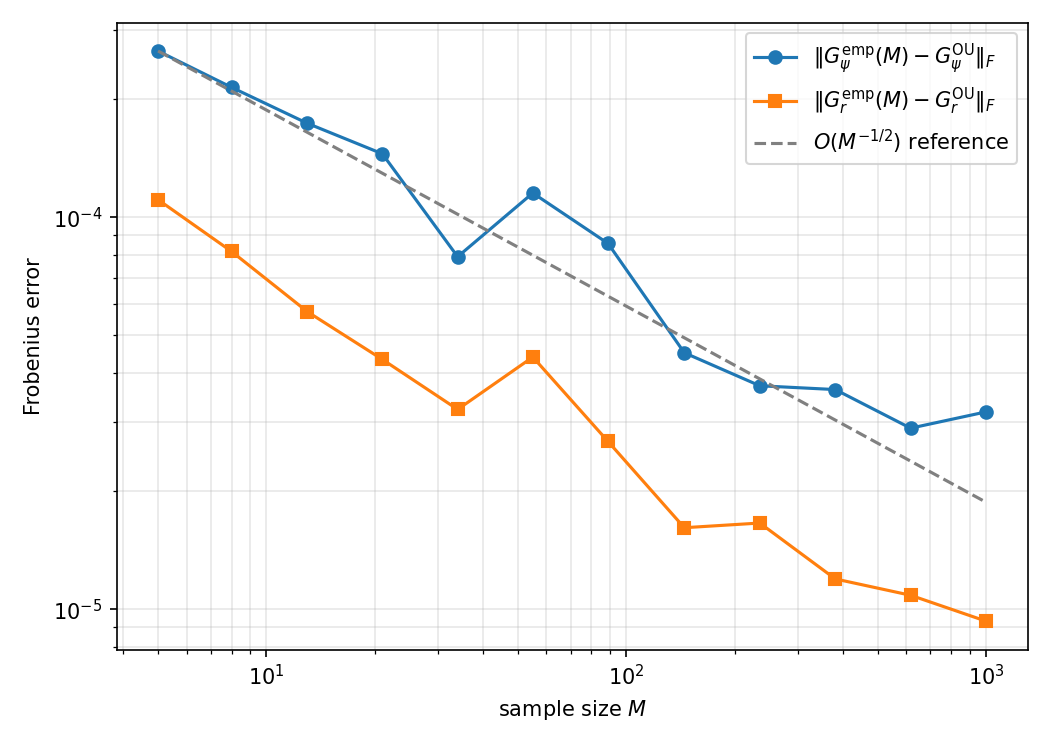}
    \caption{Convergence of the empirically estimated Gram matrices
    $G_\psi$ and $G_r$ to their closed-form OU moment targets, evaluated on the
    simulator time grid so that the comparison isolates Monte Carlo sampling
    error. The error decays at the $O(M^{-1/2})$ Monte Carlo reference rate
    (grey), consistent with consistent estimation of the OU-projected blocks;
    imposing the spread model removes this sampling noise for those blocks.}
    \label{fig:gram_convergence}
\end{figure}

The Gram convergence above shows that the OU-projected moment blocks are
estimated consistently; we now show that using their \emph{closed forms} instead
of Monte Carlo estimates produces a better-calibrated policy. We build a
projected policy on the low-dimensional basis $\psi_t=(1,t,S_t)$ and its integral
$r_t=\int_0^t\psi_u\,\dd u$, and solve the reduced quadratic programme of
\Cref{thm:quadratic_reduction_signature} twice: once with the moment blocks
estimated from $M$ simulated paths, and once with the closed-form OU blocks
$G_\psi$, $G_r$, $\E[r_Tr_T^\top]$ and the closed-form signal term derived in
\Cref{sec:level2_structure}.\footnote{This diagnostic isolates the blocks that
admit OU closed forms: it uses the impact and inventory-risk terms ($\eta=0$,
$\phi=1$, $\Sigma=I_2$) and the terminal penalty, for which $G_\psi$, $G_r$ and
$\E[r_Tr_T^\top]$ suffice; the dollar-neutrality block does not reduce to OU
Gram matrices and is left to the empirical full-signature construction.}
Both are compared against a large independent reference ($4000$ paths).

\Cref{fig:ou_bridge} reports, as a function of $M$, the relative coefficient
error, the relative objective gap, the expected terminal inventory, and the
conditioning of the projected curvature matrix. The closed-form policy attains a relative
coefficient error of $1.3\times10^{-2}$ and objective gap of $2.4\times10^{-3}$
\emph{at any $M$}, below the empirical policy even at $M=1000$
($4.8\times10^{-2}$ and $1.6\times10^{-2}$ respectively) and well below the
small-sample empirical fit ($M=25$: $9.3\times10^{-2}$ and $1.2\times10^{-1}$).
\begingroup
The projected curvature matrices have the same order of conditioning, so
the improvement is driven by variance reduction in the moment blocks rather
than by a material conditioning change, and the terminal-inventory curves agree once the empirical
estimator has converged. This is the concrete payoff of \Cref{sec:level2_structure}:
where the spread model holds, closed-form blocks deliver large-sample
calibration accuracy at no sampling cost, motivating a hybrid implementation
that replaces the noisiest empirical blocks by their analytic counterparts. The
headline $N=2$ policy of this section is still calibrated from the full empirical
signature moments; the bridge applies to the projected OU sub-block only.

\begin{figure}
    \centering
    \includegraphics[width=\linewidth]{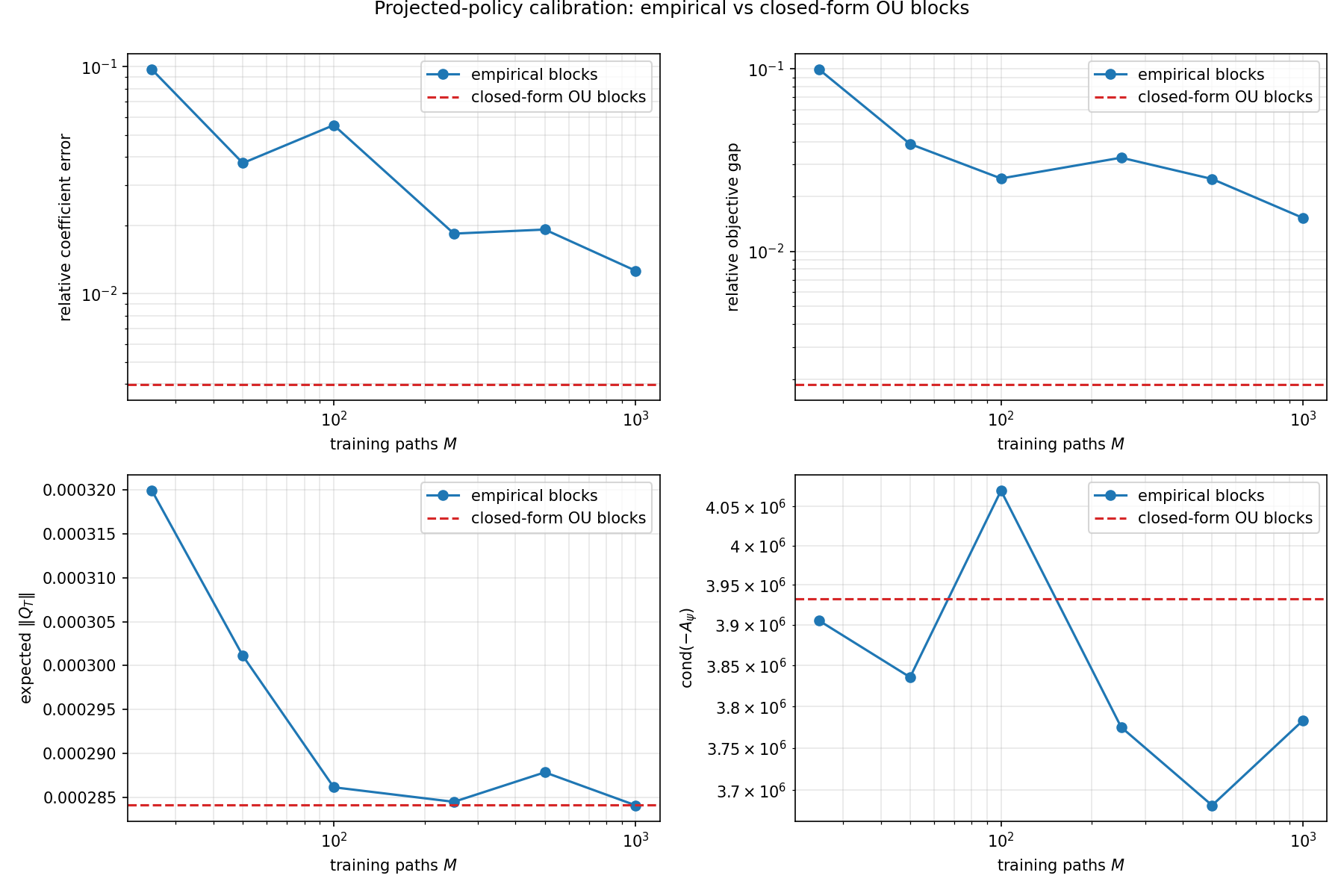}
    \caption{Projected-policy calibration on $\psi=(1,t,S_t)$: empirical
    Monte-Carlo blocks (blue) versus closed-form OU blocks (red dashed), against
    a $4000$-path reference. Closed-form blocks achieve lower relative
    coefficient error and objective gap than the empirical fit at every training
    size $M$, while the projected curvature matrices have comparable
    conditioning: the gain is variance reduction, not regularisation.}
    \label{fig:ou_bridge}
\end{figure}

\section{\texorpdfstring{Numerical Validation of the Quadratic Reduction}{Numerical Validation of the Quadratic Reduction}}
\label[appendix]{subsec:reduction_check}

In this section we verify the central claim of
\Cref{thm:quadratic_reduction_signature} numerically: that the path-dependent
objective collapses \emph{exactly} to the finite quadratic form
$J(\theta)=\theta^\top A\theta+b^\top\theta+c$. On a fixed set of $300$ paths we
assemble $(A,b)$ from the tensor formulas and, independently, recompute the
objective directly from $v_t=Bx_t$, $Q_t=Q_0+By_t$, and $\alpha_t=Kx_t$ by
integrating the running terms of \Cref{eq:objective_framework} against time.
\Cref{fig:reduction_check} plots the two values against each other for the fitted
$\theta^*$ and for $48$ random coefficient vectors spanning a wide range of
objective values. They coincide along the identity line with a maximum relative
error of $2.2\times10^{-7}$, i.e.\ to
floating-point precision. This is an algebraic identity per path rather than a
Monte Carlo statement, so the agreement checks that the tensor assembly of
$(A,b,c)$ implements the reduction correctly.

\begin{figure}
    \centering
    \includegraphics[width=0.7\linewidth]{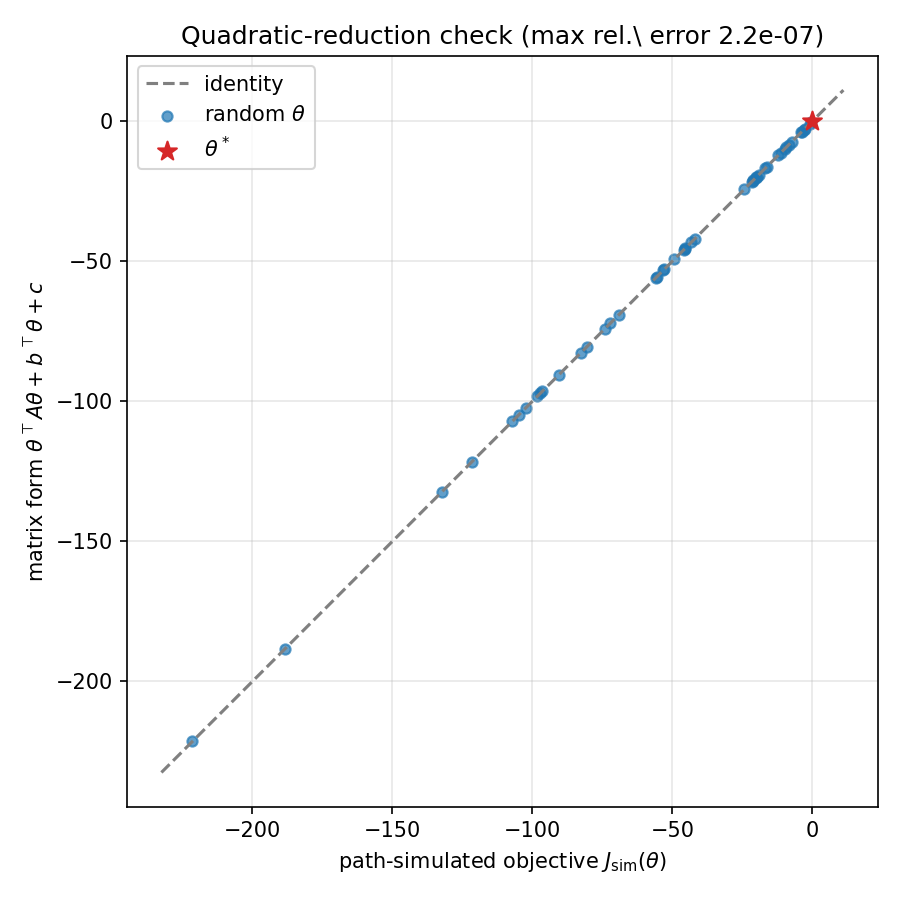}
    \caption{Quadratic-reduction check. The matrix-form value
    $\theta^\top A\theta+b^\top\theta+c$ versus the directly path-simulated
    objective $J_{\mathrm{sim}}(\theta)$, for the optimiser $\theta^*$ (star) and
    $48$ random coefficient vectors. All points lie on the identity line
    (maximum relative error $2.2\times10^{-7}$), checking the matrix
    implementation of \Cref{thm:quadratic_reduction_signature} numerically.}
    \label{fig:reduction_check}
\end{figure}

The next diagnostic, \Cref{fig:matrix_spectrum}, then asks whether the complete
quadratic problem assembled from empirical moment estimates has the sign and
conditioning properties required by the closed-form optimiser.

\subsection{\texorpdfstring{Matrix Diagnostics}{Matrix Diagnostics}}
\label{subsec:matrix_diagnostics}

\Cref{lem:theta_star} requires $A$ to be symmetric negative definite
for the closed-form maximiser $\theta^* = -\frac12 A^{-1}b$ to exist without
regularisation.
\Cref{fig:matrix_spectrum} displays the signed eigenvalue spectrum of $-\hat A$ and
heatmap of the symmetrised matrix $\hat{A}$ obtained from $1000$
training paths at signature order $N = 2$.  The spectrum of $-\hat A$ is
non-negative up to numerical precision and spans many orders of magnitude: only
29 of the 42 directions lie above the ridge level $\lambda_{\mathrm{ridge}}$,
while the remainder form a near-null tail. The ridge shift
$\lambda_{\mathrm{ridge}}=10^{-8}$ therefore makes the linear solve well posed
precisely in those near-null signature directions.
The near-null directions are expected: the level-two feature dictionary
contains correlated time-price and spread coordinates, and the zero-inventory
initial condition means that only directions rewarded by the signal and
regularised by impact or terminal liquidation matter out of sample. The
diagnostic therefore supports the algebraic sign of the objective while also
showing why a small ridge is a practical part of the calibration. The sensitivity analysis below shows that
realised metrics can be materially affected by this choice.
Together, \Cref{fig:gram_convergence} and \Cref{fig:matrix_spectrum} check the
numerical reduction itself: the former checks the moment inputs, while the
latter checks the resulting quadratic form. Only after these checks does it
make sense to interpret the trading paths and metric-specific diagnostics.

\begin{figure}
    \centering
    \includegraphics[width=\linewidth]{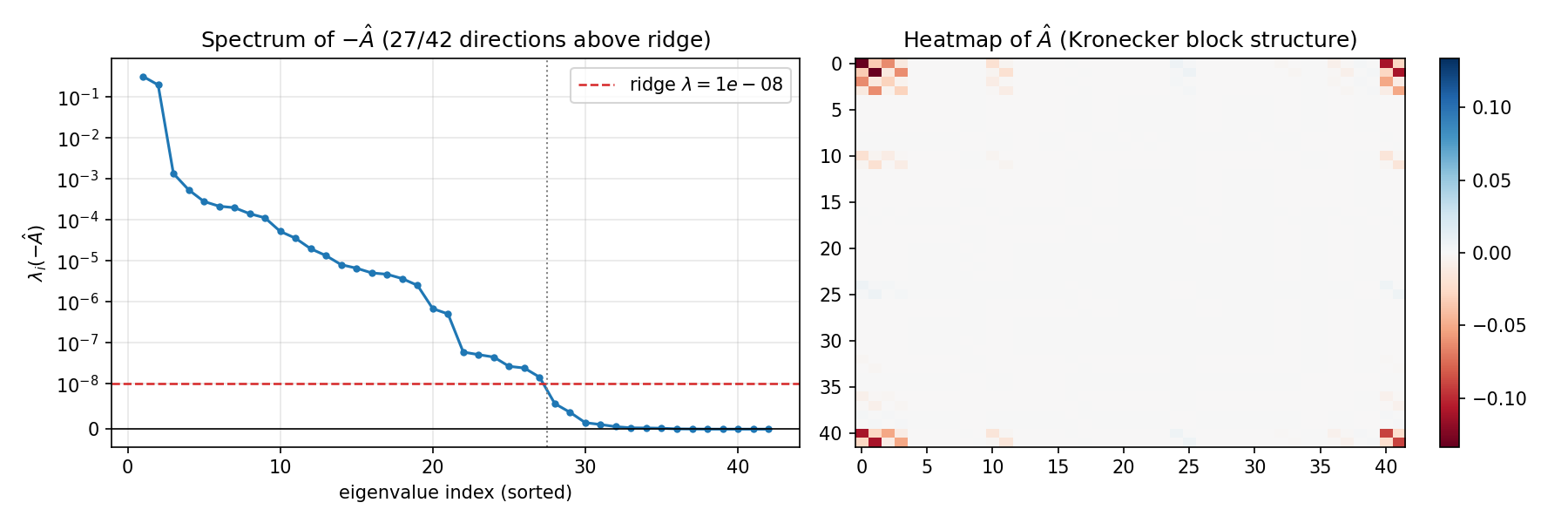}
    \caption{Left: sorted signed eigenvalues of $-\hat A$ (positive
    semidefinite under concavity) on a signed-log scale, with zero and the ridge
    level $\lambda_{\mathrm{ridge}}$ marked. Plotting signed eigenvalues keeps
    any concavity violation visible. Right: heatmap of $\hat A$ showing its Kronecker
    block structure from signature Gram matrices and execution penalties.}
    \label{fig:matrix_spectrum}
\end{figure}

\subsection{\texorpdfstring{Effect of Signature Truncation Order}{Effect of Signature Truncation Order}}
\label{subsec:sig_order}

We compare policies trained at signature orders $N = 1$ and $N = 2$.
For each order, 1000 training paths are used to estimate $A$ and $b$,
and the resulting policy is evaluated on 1000 independent test paths.
\Cref{fig:sig_order_comparison} reports the pathwise distributions of the
reduced objective, terminal wealth, terminal inventory norm, and turnover.

At $N = 1$, the admissible policy class is linear in the current
signature state: trading speed is an affine function of $(1, t, P_t^{(1)},
P_t^{(2)}, z_t)$.%
\footnote{Under the standard signature convention, the level-one coordinates
record the increments $Z_t - Z_0$ (plus the scalar term~$1$); equivalently,
these are affine functions of the current path values.}
While the $z$-score is available, the policy cannot exploit
\emph{cross-products} of features, so information about how the spread
has been evolving (its velocity, curvature, or lead-lag ordering between
the two assets) is inaccessible.
At $N = 2$, the policy gains access to the full level-two block of the
signature, which includes second-order iterated integrals such as the
L\'evy area (\Cref{rem:levy_area_signal}) and
time-weighted price moves.
These capture geometric path information that is strictly beyond level one:
for instance, whether a spread deviation was driven by a fast move in
one asset or a slow drift in both.
Under the baseline calibration used here, the level-two specification raises
both the reduced objective and the terminal mark-to-market wealth accounting
diagnostic relative to the level-one policy, at the cost of higher turnover and slightly larger residual
inventory.
Quantitatively, the mean reduced objective rises from
$4.125\times10^{-5}$ at $N=1$ to $6.438\times10^{-5}$ at $N=2$, while
mean terminal wealth increases from $4.787\times10^{-5}$ to
$2.105\times10^{-4}$. The accompanying increase in turnover, from $0.544$
to $1.204$, shows that the metric gains are obtained through a more active
policy.

\begin{figure}
    \centering
    \includegraphics[width=\linewidth]{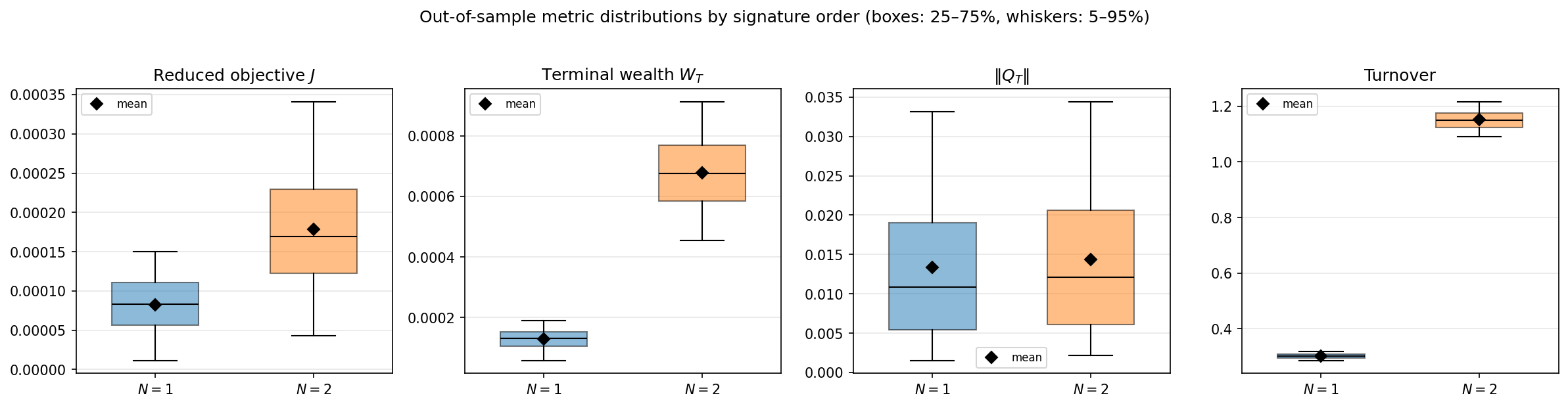}
    \caption{Effect of signature truncation order $N$. Out-of-sample per-path
    distributions (boxes: 25--75\%, whiskers: 5--95\%, diamonds: means) of the
    reduced objective, terminal wealth, terminal inventory norm, and turnover at
    $N=1$ versus $N=2$. Adding the second-order block (time ordering, L\'evy
    area) shifts the reduced-objective and terminal-wealth accounting distributions upward, while the
    turnover and inventory distributions widen markedly, due to the cost of the more
    active level-two policy.}
    \label{fig:sig_order_comparison}
\end{figure}

\subsection{\texorpdfstring{Sensitivity Analysis}{Sensitivity Analysis}}
\label{subsec:sensitivity}
We next examine how the fitted strategy reacts to changes in the
main calibration and execution parameters. The sensitivity runs use the
common-trend log-spread dynamics of \Cref{subsec:application_dgp} and start from
the baseline configuration in \Cref{eq:solver_parameters_base_set}. We vary one
parameter family at a time and keep all other parameters fixed. The parameters considered include the temporary-impact coefficients $\tilde{\Lambda}_{ij}$, the dollar-neutrality parameter $\eta$, the terminal inventory penalty $\gamma$, the inventory-risk parameter $\phi$, the ridge parameter, the signal scale $c_{\alpha}$, and the number of executed trades. For each value, the policy is recalibrated on the training paths and evaluated on the same synthetic test set.

\Cref{fig:sensitivity_sweeps} reports the resulting mean accounting return on turnover
under the canonical synthetic workflow. The purpose of this diagnostic is not
to optimise these parameters, but to identify which modelling parameters
materially affect out-of-sample accounting metrics. The model is especially
sensitive to the ridge parameter, which controls the use of near-null directions
in the empirical curvature matrix: too little regularisation can amplify noisy
signature directions, whereas too much regularisation shrinks the policy. The
operating ridge should therefore be read as a regularised empirical choice, not
as an unregularised optimum.

Overall, the sensitivity analysis supports the interpretation of the baseline experiment. The level-two reduced-objective gain is not presented as the result of a single finely tuned parameter; rather, the policy's behaviour is governed by the expected trade-off between signal strength, execution cost, neutrality, terminal liquidation, and regularisation.

\begin{figure}[h]
    \centering
    \includegraphics[width=\linewidth]{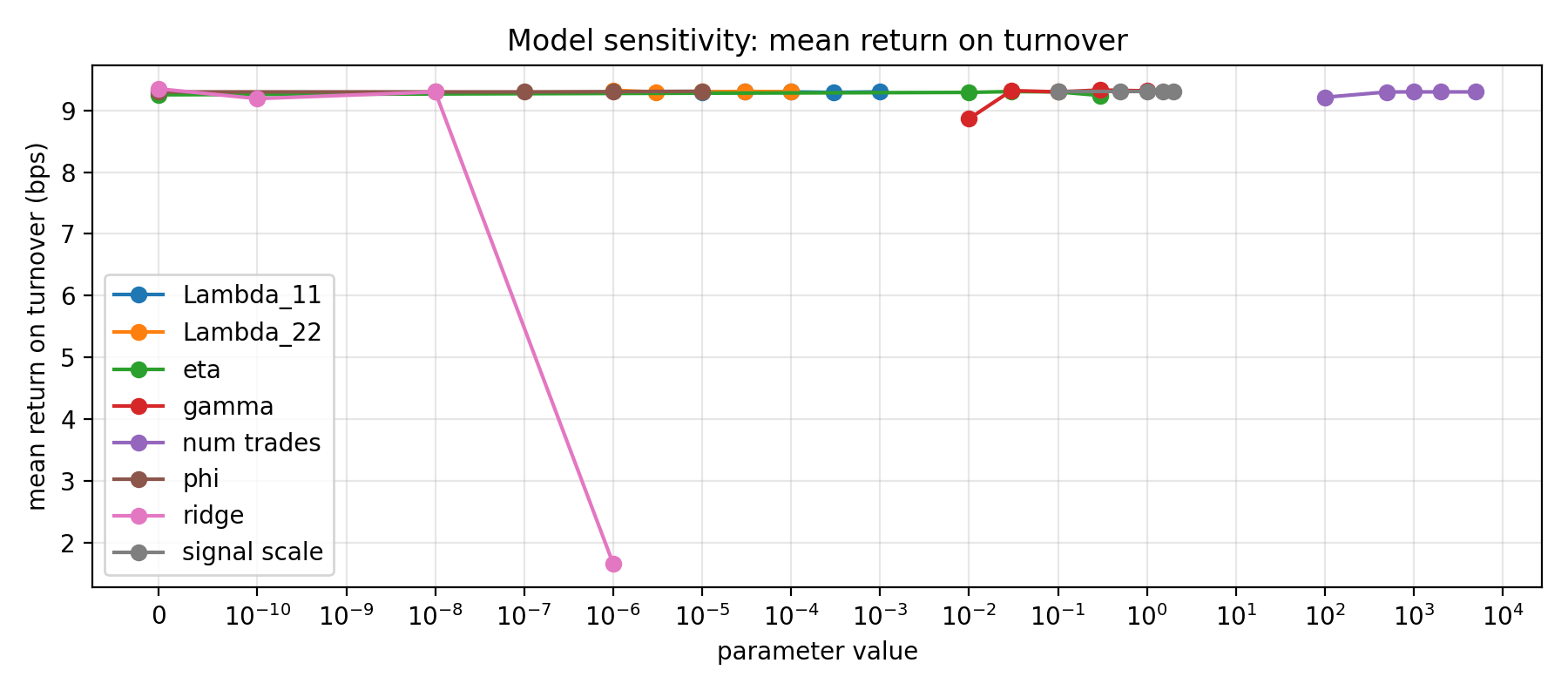}
    \caption{Sensitivity analysis for the synthetic mean-reversion experiment.
    Each curve varies one model or execution parameter while keeping the
    remaining parameters fixed at the baseline calibration set in
    \cref{subsec:sensitivity}. The reported metric is mean return on turnover,
    in basis points, computed on 1000 synthetic test paths. As expected, the accounting
    metric is especially sensitive to the ridge parameter, indicating that too
    little regularisation can amplify noisy signature directions whereas too
    much regularisation shrinks the policy. The baseline calibration remains in
    a stable region for the remaining parameters.}
    \label{fig:sensitivity_sweeps}
\end{figure}

\endgroup

\section*{Disclosure of interest}

The authors report that there are no competing interests to declare.

\section*{Funding}

No funding was received for this work.

\end{document}